%% file: badia_hampton_principe_parallel_mlmc_2021_last.tex
\documentclass[reqno,oneside,11pt]{amsart}
\usepackage[foot]{amsaddr}
\usepackage[T1]{fontenc}
\usepackage[utf8]{inputenc}
\usepackage[top=1.5cm,bottom=2cm,right=2cm,left=2cm]{geometry}
\usepackage{hyperref} 
\usepackage{amsmath,amsfonts,amssymb,amscd,amsthm}
\usepackage[ruled, lined, longend, linesnumbered]{algorithm2e}
\usepackage[usenames,dvipsnames,svgnames]{xcolor}
\usepackage[numbers]{natbib}
\usepackage{doi}
\usepackage{caption}
\usepackage{subcaption}
\usepackage{autonum} 
\usepackage{mathtools}
\usepackage[hang]{footmisc}
\usepackage{newtxtext}
\usepackage{newtxmath}
\usepackage{acronym}

\newcommand{\theauthors}{Santiago Badia,Jerrad Hampton and Javier Principe}
\newcommand{\thetitle}{A Massively Parallel Implementation of Multilevel Monte Carlo for Finite Element Models}
\hypersetup{
    bookmarks=true,
    breaklinks=true,
    bookmarksopen=true,
    pdftitle={\thetitle},  
    pdfauthor={\theauthors}, 
    colorlinks=true,       
    linkcolor=black,       
    citecolor=blue,        
    filecolor=black,       
    urlcolor=blue          
}

\usepackage{booktabs}

\usepackage{tikz}
\definecolor{myellow}{RGB}{255,230,128}
\definecolor{gray20}{RGB}{204,204,204}
\definecolor{mygray}{RGB}{204,204,204}
\definecolor{mygreen}{RGB}{138,203,95}
\definecolor{lgreen}{RGB}{34,102,0}
\definecolor{myblue}{RGB}{77,151,214}

\usepackage[textsize=small]{todonotes}

\usepackage{soul}

\input{AlgorithmStyle}

\SetAlgoCaptionLayout{SetSmall}

\newtheorem{theorem}{Theorem}[section]

\newtheorem{remark}[theorem]{Remark}

\tikzstyle{fig-ph}=[draw,minimum width=\textwidth, minimum height=\textwidth,text width=0.9\textwidth,color=red]

\tikzstyle{fig-ph-r}=[draw,minimum width=\textwidth, minimum height=1.4\textwidth,text width=0.9\textwidth,color=red]
\tikzstyle{fig-ph-rl}=[draw,minimum width=\textwidth, minimum height=0.5\textwidth,text width=0.9\textwidth,color=red]

\acrodef{dof}[DOF]{Degree Of Freedom}
\acrodef{mpi}[MPI]{Message Passing Interface}
\acrodef{uq}[UQ]{Uncertainty Quantification}
\acrodef{pde}[PDE]{Partial Differential Equation}
\acrodef{agfem}[AgFEM]{Aggregated Finite Element Method}
\acrodef{agfe}[AgFE]{Aggregated Finite Element}
\acrodef{fe}[FE]{Finite Element}
\acrodef{mc}[MC]{Monte Carlo}
\acrodef{qmc}[QMC]{Quasi Monte Carlo}
\acrodef{mlmc}[MLMC]{Multilevel Monte Carlo}
\acrodef{amlmc}[AMLMC]{Adaptive Multilevel Monte Carlo}
\acrodef{hpc}[HPC]{High Performance Computing}
\acrodef{qoi}[QoI]{Quantity of Interest}
\acrodef{amg}[AMG]{Algebraic Multigrid}
\acrodef{prng}[PRNG]{Pseudorandom Number Generator}
\acrodef{acr}[ACR]{Active Computation Ratio}
\acrodef{ncr}[NCR]{Non-Computation Ratio}
\acrodef{petsc}[PETSc]{Portable, Extensible Toolkit for Scientific Computation}
\acrodef{dd}[DD]{Domain Decomposition}
\acrodef{cg}[CG]{Conjugate Gradient}
\acrodef{spmd}[SPMD]{single-program multiple-data}
\acrodef{KLE}[KLE]{Karhunen-Loève Expansion}

\def\rev#1{{ #1}}
\def\rerev#1{{ #1}}
\def\rererev#1{{ #1}}
\DeclareTextFontCommand{\vbt}{\ttfamily\hyphenchar\font=45\relax}
\def\FEMPAR{{\vbt{FEMPAR}} }
\def\grad{{\boldsymbol{\nabla}}}
\def\diff{{\kappa}}

\def\x{{\boldsymbol{x}}}
\def\y{{\boldsymbol{y}}}

\def\w{{\omega}}

\def\D{{\mathcal{D}}}
\def\M{{\mathcal{M}}}

\def\lev{{\ell}}

\def\artdom{\mathcal{B}}

\font\doble=msbm10 scaled\magstep1

\newcommand\E{\hbox{\doble E}}
\newcommand\V{\hbox{\doble V}}

\def\T{\mathcal{T}}
\def\Th#1{Theorem~\ref{th:#1}}
\def\Rm#1{Remark~\ref{rm:#1}}
\def\Eq#1{~(\ref{eq:#1})}
\def\Fig#1{Fig.~\ref{fig:#1}}
\def\Tab#1{Table~\ref{tab:#1}}
\def\Alg#1{Alg.~\ref{alg:#1}}
\def\Line#1{line~\ref{line:#1}}
\def\Sec#1{Section~\ref{sec:#1}}

\graphicspath{{}{}{}{}{}}

\begin{document}

\title{A Massively Parallel Implementation of Multilevel Monte Carlo for Finite Element Models}
\author[S. Badia]{Santiago Badia$^{1,2}$}
\author[J. Hampton]{Jerrad Hampton$^{2}$}
\author[J. Principe]{Javier Principe$^{3,2,*}$}

\thanks{\null\\
$^1$ School of Mathematics, Monash University, Clayton, Victoria, 3800, Australia.\\
$^2$ Centre Internacional de M\`etodes Num\`erics en Enginyeria, Esteve Terrades 5, E-08860 Castelldefels, Spain.\\
$^3$ Universitat Polit\`ecnica de Catalunya, Campus Diagonal Bes\`os, Av. Eduard Maristany 16, Edifici A (EEBE), 08019, Barcelona, Spain\\
$^*$ Corresponding author.\\
E-mails: {\tt santiago.badia@monash.edu} (SB), {\tt jhampton@cimne.upc.edu} (JH), {\tt principe@cimne.upc.edu} (JP)}

\thanks{This research was supported by the European Union's Horizon 2020 research and innovation programme under the ExaQUte project, with grant agreement No 800898, the project RTI2018-096898-B-I00 from the ``FEDER/Ministerio de Ciencia e Innovaci\'{o}n - Agencia Estatal de Investigaci\'{o}n'' and the Australian Government through the Australian Research Council (project number DP210103092). The authors also acknowledge the Severo Ochoa Centre of Excellence (2019-2023), which financially supported this work
under the grant CEX2018-000797-S funded by MCIN/AEI/10.13039/501100011033. The authors thankfully acknowledge the computer resources at MareNostrum and the technical support provided by Barcelona Supercomputing Center (IM-2019-3-0012). JH has received funding from the European Union's Horizon 2020 research and innovation programme under the Marie Sk\l{}odowska-Curie grant agreement No 712949 (TECNIOspring PLUS) and from the Agency for Business Competitiveness of the Government of Catalonia.}
\date{January 30, 2023}
\begin{abstract}
  The \ac{mlmc} method has proven to be an effective variance-reduction statistical method for \ac{uq} in \ac{pde} models, combining model computations at different levels to create an accurate estimate. Still, the computational complexity of the resulting method is extremely high, particularly for 3D models, which requires advanced algorithms for the efficient exploitation of \ac{hpc}.
  In this article we present a new implementation of the \ac{mlmc} in massively parallel computer architectures, exploiting parallelism within and between each level of the hierarchy. The numerical approximation of the \ac{pde} is performed using the finite element method but the algorithm is quite general and could be applied to other discretization methods.
  The two key ingredients of the implementation are a good processor partition scheme together with a good scheduling algorithm to assign work to different processors. We introduce a multiple partition of the set of processors that permits the simultaneous execution of different levels and we develop a dynamic scheduling algorithm to exploit it.
  The problem of finding the optimal scheduling of distributed tasks in a parallel computer is an NP-complete problem. We propose and analyze a new greedy scheduling algorithm to assign samples and we show that it is a 2-approximation, which is the best that may be expected under general assumptions. On top of this result we design a distributed memory implementation using the \ac{mpi} standard. Finally we present a set of numerical experiments illustrating its scalability properties.
\end{abstract}


\maketitle

\input{content.tex}

\bibliographystyle{plain}
\bibliography{refs.bib}
\end{document}

%% file: AlgorithmStyle.tex
\SetAlFnt{\footnotesize}
\SetArgSty{textnormal}

\SetAlCapSty{xAlCapSty}

\SetNlSty{mynlfont}{}{} 

\LinesNumbered

\SetSideCommentRight

\DontPrintSemicolon

\RestyleAlgo{algoruled}

\SetInd{0.5em}{0.5em}

%% file: content.tex
\section{Introduction}
\label{sec:intro}

\ac{uq} requires the solution of stochastic \acp{pde} with random data. Some methods for solving stochastic \acp{pde}, e.g. stochastic Galerkin \cite{lemaitre2010spectral,ghanem1991stochastic}, are based on a standard approximation in space like \acp{fe} or finite volumes, and different types of polynomial expansions in the stochastic direction \cite{Xiu2003}. Although generally powerful, these techniques typically suffer two significant drawbacks. The first one is being intrusive, i.e. a code that can be used to solve a deterministic problem must be modified to solve the stochastic analogue. The second one is the poor scaling in the dimension of the stochastic space, i.e. suffering from the so-called ``curse of dimensionality'' \cite{Babuska2010}.

In contrast, sampling methods, i.e. \ac{mc} and its variants, have a convergence rate which is independent of the stochastic dimension and are non-intrusive. The only necessary assumption for such methods to converge to the exact statistics of the solution when the number of input samples tends to infinity is the existence of finite second moments, i.e. finite variance. This assumption is easily satisfied for many physical systems, making the method generally applicable to many problems of practical interest. They require the repeated evaluation of a deterministic model with different randomly generated inputs and can be implemented in a way that does not require intrusive modifications of the deterministic code. However, the number of samples required to achieve statistical convergence combined with the complexity of the computational model required to have enough spatial and/or temporal accuracy can make the \ac{mc} method rather expensive. Indeed, the reduction of overall computational cost to achieve a desired error tolerance is the principal motivation for the development of its variants, like the \ac{mlmc} method considered herein.

Stochastic collocation methods \cite{Babuska2010} are also based on polynomial expansions but using orthogonal basis and appropriate quadratures. The final computation requires independent evaluations of the deterministic model, in the same way as sampling methods, thus sharing the non-intrusiveness property. Appropriate quadratures are required to mitigate the curse of dimensionality, which is a field of intense research \cite{Xiu2005,Chen2017}. Similarly, importance sampling methods \cite{Rauhut2012, Hampton2015, Adcock2017}, and the deterministic design of experiments \cite{Diaz2018b} are also known to reduce the necessary number of model evaluations, particularly in the contexts of function approximation \cite{Luthen2021}. The algorithms discussed herein can be extended for parallelization of these ``smart sampling'' methods as well, but the required modifications are beyond the scope of this article.

Focus on this important complexity reduction makes \ac{mlmc} \cite{Kebaier2005, Giles2008, Giles2012, Mishra2012, Mishra2012a, Pisaroni2017}, and its adaptive variants \cite{Giles2009,Cliffe2011,Elfverson,Collier2015} arguably the most practical extensions of \ac{mc}. These methods rely on different levels of computational effort for the deterministic model, e.g. a hierarchy of spatial or temporal meshes. These methods seek to combine samples on each level to benefit from the high accuracy of the expensive, higher levels, and the low computational cost of the less accurate, lower levels. A key aspect of the algorithm is the evaluation of how many samples should be used at each level.

Sampling methods for \ac{uq}, including \ac{mlmc}, belong to a general class of methods for outer-loop applications \cite[Section 1.2]{Peherstorfer2018}, defined as computational applications in which there is an outer loop around a model that is called at each iteration to evaluate a function. This includes optimization and statistical inference apart from \ac{uq}. For example, in the case of optimization under uncertainty, a \ac{uq} problem is solved as an inner loop within each step of the outer loop of the optimization algorithm \cite{de2020touu}. A big effort was performed during the 90s for the development of the Dakota library \cite{Eldred1998,Eldred2000} aimed to deal with this class of problems. Parallelization opportunities where classified into two areas, algorithmic and function evaluation\footnote{The classification in \cite{Eldred1998,Eldred2000} actually included four categories but two are enough for the current discussion about \ac{mlmc}.} and it was already observed that the former requires very little or no inter-processor communication in contrast to the latter. Therefore, the first question posed was how to select the amount of parallelism used on each area, i.e. one has to decide between the assignment of processors to the parallelization of the model or the parallelization of the outer loop (with the concurrent execution of several models). The analysis in  \cite{Eldred1998,Eldred2000} shows that the most efficient choice is to use the minimum number of processors that permit to run the model, for the simple reason that some loss of efficiency occurs in strong scaling. 

The parallelization of \ac{mlmc} has been considered in recent years \cite{Sukys2014,Sukys2014a,Gmeiner2016,Shegunov2020,Zakharov2020,Gantner2016,Baumgarten2021} as multiple levels of spatial/temporal discretization introduce additional scheduling challenges to effectively utilize resources of a parallel environment with minimal idling.  One can conceptually exploit three levels of parallelism in \ac{mlmc}. 
The first is the parallelization of the (deterministic) \ac{pde} solver itself, which is often dependent on the level, meaning that different levels will use different computational resources. The second is the parallelization of the different \ac{mc} samples on each level of \ac{mlmc}, which is the easier one given the independence between them and the large number required, especially at lower discretization levels. The third one is the parallelization between levels, which we found important for the optimality of the algorithm, as discussed below. The two main ingredients of the implementation are i) a processor partitioning scheme and ii) a scheduling strategy and the goal is to maximize efficiency.

Scheduling strategies for \ac{mlmc} sampling can be classified \cite{Gmeiner2016} in three different ways. First, they are classified according to the number of parallelization layers exploited simultaneously, inside samples, between samples and between levels. Second, they are classified as homogeneous, when all the processors are working to sample on the same discretization level with the same number of processors per sample, and heterogeneous otherwise. It is important to note that, according to this definition, heterogeneous  scheduling includes many different possibilities, e.g. using different number of processors for running samples at the same level (which is considered in \cite{Gmeiner2016}) or running samples at different levels simultaneously (which is {\it not} considered in \cite{Gmeiner2016}). A final classification is between static scheduling, when the assignment of work is done before starting to sample, and dynamic scheduling, where assignment is made ``on the fly''. It is also important to note that the term ``dynamic'' is used in \cite{Gmeiner2016} to refer to scheduling strategies where the work is assigned depending on a cost estimation made at the beginning of the calculation but {\it before} starting to sample. Several scheduling algorithms are studied in \cite{Gmeiner2016}, mostly homogeneous, but also a couple of heterogeneous variants (with and without strong scaling), both static and dynamic, although three layer parallelism was not investigated.

A scheduling with three layer parallelism was proposed in \cite{Sukys2012,Sukys2014,Sukys2014a} with two variants, a static one in which the distribution is made based on work estimates given by the size of the discretization at each level and one in which the distribution of work depends also on the \ac{pde} coefficients (which are sample dependent) and is performed at execution time but before starting to sample. In these references, ``adaptive load balancing'' is used to refer to this second algorithm and it is explicitly stated that it is not ``dynamic''  see \cite[page 111]{Sukys2014a}, although it is categorized in this way in \cite{Gmeiner2016}. The processor partitioning scheme in \cite{Sukys2012,Sukys2014,Sukys2014a} divides available units into levels, each of them composed by samplers that contain several cores each, that is, the computational power is distributed into {\it many} levels and samples at a given time and assignment is made depending on the load balancing strategy just mentioned. For this scheme a 2-approximation is proved in \cite{Sukys2014,Sukys2014a} assuming strong scaling of the \ac{pde} solver, that is, the execution time is at most twice the optimal. The observed efficiency is improved by the adaptive load balancing algorithm although actual computational times are slower than in the static case \cite[Figs. 3 and 4]{Sukys2014}, which shows a very good strong and weak scaling up to 40k cores. A similar strategy was followed in \cite{Gantner2016}, where an object oriented implementation of several sampling methods, including \ac{mc} and \ac{qmc}, either single or multilevel is presented. 

Another dynamic scheduling strategy for \ac{mlmc} was recently proposed in \cite{Tosi2021}. It exploits general purpose scheduling algorithms in \cite{Tejedor2017} which permit to deal with complex dependencies between tasks and this is exploited to extract parallelism between samples and levels. Although in principle it can be also used with parallel sampling, the mapping of tasks to processors and its effects on load balancing is not discussed. Moreover, parallelization at the sample level is not shown in the numerical examples.

In this article we propose \emph{the first dynamic scheduling strategy with three layer parallelism} that works on top of a multiple partition of the set of available processors. Given the number of levels of the \ac{mlmc} algorithm and assuming a weak scaling of the \ac{pde} solver we define the number of processors per sample in terms of the number of processors used in the coarsest one, as in \cite{Gmeiner2016}. We then consider multiple partitions of the set of available processors by these numbers of processors, thus obtaining $L$ different partitions, where $L$ is the maximum number of levels of the \ac{mlmc} method. Together with this multiple partition we present a dynamic scheduling algorithm that assigns tasks to available resources prioritizing finer levels but allowing the concurrent execution of samples at different levels. This distinctive feature of our implementation admits a proof that the algorithm is a 2-approximation, with a very small idling actually occurring only at the end of the computation or at the end of an adaptive \ac{mlmc} iteration. Besides, excellent scalability is obtained, both strong and weak. With some additional modifications discussed below, this scalability is observed even for the difficult case of models having low per-sample evaluation times. Therefore, the novelty of this approach includes
\begin{itemize}
\item a dynamic scheduling algorithm that exploits three layer parallelism: for each individual sample, across samples, and across levels;
\item the description of an \ac{mpi} implementation based on a master-slave strategy with a multiple partition of slave processors;
\item a simple yet robust, dynamic batch sampling mechanism to reduce the communication between the master and the slaves;
\item a parallelization of the master for eliminating a bottleneck at extreme scales;
\item the numerical demonstration of the scalability of the implementation, also when stressed by short sampling times.
\end{itemize}

The article is structured as follows. \Sec{mlmc} describes the \ac{mlmc} method and an adaptive variant implemented herein and includes the description of the \ac{pde} we consider as an example. \Sec{scheduling} describes the scheduling in an abstract way and its 2-approximation property. \Sec{implementation} describes the actual implementation,  including the parallel partition strategy. \Sec{examples} presents several examples to demonstrate the weak and strong scalability of the implementation efficiency.

\section{The multilevel Monte Carlo method}
\label{sec:mlmc}

The principal idea of \ac{mlmc} is to exploit a hierarchy of model discretizations to reduce the overall computational cost by transferring the majority of the sampling cost to the cheaper models, and having an accuracy governed by the more expensive models. The method is appropriate for any hierarchy of models for which convergence of the \ac{qoi} is known, although such knowledge is not generally necessary.

\subsection{Model problem}
\label{sec:problem}
In this work we consider the following elliptic stochastic problem, although the algorithm permits consideration of a general class of \ac{pde} problems using parallel solvers. Given an oriented manifold $\M(\w)\subset \mathbb{R}^d$ and its corresponding interior domain $\D(\w)$, find $u(\x,\w)$ such that
\begin{align}
  -\grad \cdot \left( \diff \grad u \right) = f  \ \hbox{ in } \ \D(\w), \qquad
  u = u_0 \ \hbox{ on } \, \M(\w), \label{eq:model}
\end{align}
where $\w \in \Omega $, denotes the uncertainty, described by a complete probability space. Additionally, the stochastic coefficients, i.e. the diffusion $\diff = \diff(\x,\w)$, forcing term $f=f(\x,\w)$, and boundary condition $u_0 = u_0(\x,\w)$ may be considered random fields too.  We note that it is usual in the literature to assume these to be deterministic when stochastic domains are considered \cite{Chaudhry20181127, Mohan2011874, Xiu2006, Harbrecht2008, Harbrecht2013, Harbrecht2016, Dambrine2017943, Dambrine2016921}. Let us assume that the realizations of $\M(\w)$ and $\D(\w)$ are bounded almost surely. We can define a bounded artificial domain $\artdom$ that contains all possible realizations of $\M(\w)$, i.e. $\D(\w) \subset \artdom, \, \forall \w \in \Omega$. We also assume that $\diff$ and $f$ are defined in $\artdom$, independently of $\w \in \Omega$. With the random solution of this problem at hand we aim to compute $\E\left(Q(u)\right)$ where $Q$ is a \ac{qoi}, e.g. an integral of $u$ on a sub-region or surface, or an evaluation of $u$ at a specific point $\x$, and $\E$ the expectation.

The problem defined by \Eq{model} is well-posed under the assumption of uniform ellipticity  \cite{Barth2011,Babuska2010,Chaudhry20181127}, i.e. there exists $\diff_0$ such that $\diff(\x) \ge \diff_0, \, \forall \x \in \artdom$ (and $\forall \w \in \Omega$ if a stochastic diffusion $\diff(\x,\w)$ is considered). Under these assumptions, the bilinear form associated to the weak form of \Eq{model} is bounded and coercive and the Lax-Milgram lemma guarantees a solution for any $\w \in \Omega$ uniformly bounded by $ \|f\|_{L^2( \artdom)}$ and $\| u_0 \|_{H^{1/2}(\M(\w))}$ \cite{Chaudhry20181127}.

\def\artdom{\mathcal{B}}
The discretization of \Eq{model} is constructed on top of a grid $\T$, introducing a $\mathcal{C}^0$ Lagrangian \ac{fe} space. When the domain is stochastic we use the \ac{agfem} method in which the grid $\T$
is a shape regular partition of a background domain $\artdom$ which is typically a bounding box. A discrete approximation is then constructed identifying cut cells of the background mesh and building a sub-triangulation on each of them. This construction permits the integration of the weak form of the problem and a judicious choice of the degrees of freedom can obtain a well posed-problem. The \ac{agfem} method was introduced in \cite{Badia2017a}, implemented in parallel in \cite{Verdugo2019} and exploited to perform \ac{uq} in random domains in \cite{Badia2021}. The reader is referred to these publications for further details.

\subsection{The standard \ac{mlmc} method}
\label{sec:smlmc}

The \ac{mlmc} method distributes sampling on a hierarchy of discretizations to reduce the overall computational cost, with respect to that of sampling on the finest one, while keeping a similar accuracy. The hierarchy consists of $L+1$ meshes $\T_0,\T_1,...,\T_L$ of sizes $h_0>h_1>...>h_L$. 
\rerev{In particular, we consider $h_l=h_0s^{-l}$ (i.e. each mesh in the hierarchy is obtained by uniformly dividing each cell into $s^d$ subcells) with $s=2$.}
Performing $\{N_\lev\}_{\lev=0}^L$ simulations for different values of the random parameters $\w$ on $\{\T_\lev\}_{\lev=0}^L$, the expectation is corrected using the whole hierarchy as
\begin{equation}
\E(Q_L) = \E(Q_0) + \sum_{\lev=1}^L \E(Q_\lev-Q_{\lev-1}) \approx \sum_{\lev=0}^L \overline{Y_\lev}:=\widetilde{Q}_L, \label{eq:mlmc_average}
\end{equation}
where $Q_\lev=Q(u_{\lev})$ is the approximation of the \ac{qoi} computed using the \ac{fe} solution $u_\lev$ for the discretization $\T_\lev$ and $Y_\lev=Q_\lev-Q_{\lev-1}$ for $\lev=1,...,L$ and $Y_0=Q_0$. The total error of this approximation is now given by
\begin{equation}
  e^2( \widetilde{Q}_L ) = \E\left[( \widetilde{Q}_L - \E(Q) )^2\right] =  (\E(Q_L -Q))^2 + \V(\widetilde{Q}_L).
\label{eq:mlmc_error}
\end{equation}
The first term, the (squared) discretization error, is reduced by enforcing that $\T_L$ provides a sufficiently accurate approximation. Assuming that it decays as $ch^{\alpha}$ for some constant $c$ and rate $\alpha$ the maximum level required to make it smaller than, e.g. $\varepsilon^2/2$, is
\begin{equation}
  L = \lceil {\frac{1}{\alpha} \log_s \left(\sqrt{2}c \varepsilon^{-1}\right)} \rceil.
 \label{eq:mlmc_num_levels}
\end{equation}
\rererev{Theoretical estimates of $\alpha$ in terms of the regularity of the \ac{pde} and the properties of the discretization are available in some cases, e.g. for the linear finite element approximation of elliptic stochastic \acp{pde} $\alpha=2$, but $c$ and $\alpha$ are generally unknown and can only be estimated
a poseriori \cite{Elfverson}.}

The second term in \Eq{mlmc_error} is the statistical error, given by $\V(\widetilde{Q}_L) = \sum_{\lev=0}^{L} N_\lev^{-1} \V(Y_\lev)$. Under general conditions $\V(Y_\lev)$ tends to zero as $\lev\rightarrow\infty$ and in this sense, the \ac{mlmc} method can be understood as a variance reduction method.

The total computational cost is given by  $C_{\rm MLMC} = \sum_{\lev=0}^{L} N_\lev C_\lev$ where $C_\lev$ is the complexity (computational cost) of evaluating $Y_\lev$, i.e. the average cost to evaluate one sample at level $\lev$ and another one at level $\lev-1$. Its minimization allows the optimal number of samples to be taken on each level \cite{Giles2008} to have a mean squared error smaller than $\varepsilon^2/2$, as
\begin{equation}
  N_\lev =   \left\lceil 2 \varepsilon^{-2} \sqrt{\frac{\V(Y_\lev)}{C_\lev}}  \sum_{i=0}^{L} \sqrt{\V(Y_i)C_i} \right\rceil.
\label{eq:mlmc_num_samples}
\end{equation}
Observe that the number of samples per level in \Eq{mlmc_num_samples} depends on $\V(\widetilde{Q}_\lev)$, and $C_\lev$ and for some problems a theoretical estimate in terms of $h_\lev$ is possible. These estimates, however, contain unknown constants which need to be determined at an initial screening phase. An alternative is to update them during the execution, as described in the next section.

\subsection{An adaptive \ac{mlmc} method}
\label{sec:amlmc}

A practical extension of \ac{mlmc} is given by \ac{amlmc} \cite{Giles2009, Cliffe2011, Elfverson} \rev{or a sophisticated variant named Continuation Multilevel Monte Carlo \cite{Collier2015}}, which dynamically enforce the discretization and/or sampling error to be below a given tolerance. This is done by updating the maximum level $L$ and the sample sizes $\{N_0,\cdots,N_L\}$ with (ideally) optimal values that reduce the \ac{mlmc} error with minimal increase in computational cost. Since the actual expectation, variance, and cost are unknown, they are estimated from sample averages and variances, which we refer to as moment estimates. Extrapolation of these moments estimates are also needed when they are not available.

The algorithm used here begins with a simple initial set of samples, which may be chosen by any number of considerations. It is presented in \Sec{implementation}, consisting of a loop in which a set of samples are evaluated, convergence is checked and, if it is not reached, the maximum level $L$ and the sample sizes $\{N_0,\cdots,N_L\}$ are updated. In order to update $L$ an estimation of the constants $c$ and $\alpha$ in \Eq{mlmc_num_levels} is made using $\overline{Y_L}$ to estimate the discretization error \cite{Elfverson}. In order to update $N_\lev$ an estimation of variance and cost is required to evaluate \Eq{mlmc_num_samples}. Variance is estimated by the sample variance, i.e.
\begin{equation}
   V(Y_\lev) = \frac{1}{N_\lev} \sum_{i=1}^{N_\lev} (Y^i_\lev -  \overline{Y_\lev})^2.
  \label{eq:sample_variance}
\end{equation}
Cost estimates are computed for each level, by adding the total cpu-time used to compute each $Y^i_\lev$, denoted here as $C^i_\lev$ and computing the average cost $\overline{C}_\lev$ on that level. From \rev{these variance and cost estimates the number of samples needed to achieve the desired error are updated from \Eq{mlmc_num_samples} possibly with an additional fraction to ensure achieving the desired error without to avoid unnecessary outer loop iterations which reduce the efficiency of the algorithm.} 
It is important to observe that each iteration of the adaptive algorithm, in which the number of samples and the maximum number of levels are updated, requires a reduction (to compute averages and sample variances). In practice, it results in a synchronization, as described in \Sec{implementation} \rev{which makes the \ac{amlmc} more challenging to scale to large core counts than the standard \ac{mlmc}}.

\section{Parallel scheduling}
\label{sec:scheduling}

In this section we present a new scheduling algorithm for the scheduling of sampling tasks required by the \ac{mlmc} method in a parallel environment. \rev{As mentioned at the end of the previous section, the \ac{amlmc} algorithm iteratively updates the number of samples and levels, executing the \ac{mlmc} algorithm at each iteration, see \Rm{indep}.}

For a general introduction to the subject see \cite{Blazewicz2019,Vazirani2003}. In general (and imprecise) terms, the scheduling problem is the development of an algorithm for the assignment of tasks to processors to optimize some cost function, e.g. the total makespan (the maximum run-time). An important distinction to be made is between preemptive and non-preemptive scheduling. In the first case, the scheduling algorithm is developed under the assumption that tasks can be preempted and restarted later on at no extra cost. We do not consider this possibility in our implementation and we therefore concentrate only on the second case.

Some results are available under the assumption that each task is executed in one processor, as in \cite{Tosi2021}. Given a set of independent tasks $T_1,...,T_n$, whose execution times are $w_1,...,w_n$ respectively, and a set of $p$ processors, the makespan $W$ of a scheduling is bounded by the average work of each processor $\overline{W}=W_{\rm T}/p$ where
$$
W_{\rm T}=\sum_{i=1}^n w_i,
$$
and the maximum work of a task is given by
$$
W_{\rm LB}=\max \left\{ \overline{W},\max_i\{w_i\} \right\}.
$$
The bound $W\ge W_{\rm LB}$ is tight, the equality with the first argument of the maximum occurs when $w_i$ are all equal and with the second when $p<n$. This problem is NP-hard \cite{Blazewicz2019}.

A simple greedy algorithm to solve this problem is to assign a task to the processor that has the least amount of work already assigned, regardless of the execution time of each task. In this case, it is easy to show that $W<2W_{\rm opt}$, where $W_{\rm opt}$ is the makespan of the optimal scheduling \cite{Vazirani2003}, and it is said to be a 2-approximation. Moreover, the bound is tight, which is shown by considering the case of $p^2$ tasks of equal processing time followed by a single task whose processing time is $p$ times larger. This example suggests that executing the longer tasks first reduces the makespan. This fact was exploited in \cite{Graham1969} to develop a greedy algorithm in which tasks are ordered in decreasing execution time and then assigned to the first available processor, see algorithm 5.1.2 in \cite{Blazewicz2019}. The makespan bound for this algorithm is $W\le (\frac{4}{3}-\frac{1}{3p}) W_{\rm opt}$, that is, a 4/3-approximation. \rev{The extension of this result to case of parallel execution of tasks is far from obvious.}

Besides, it is important to emphasize that the execution time of a given sample in the \ac{mlmc} method is hardly known in advance.  In some cases the execution time can be estimated before the calculation based on the \ac{pde} coefficients. For example, in \cite{Sukys2014}, first order hyperbolic problems are approximated using explicit time integration methods with a time step determined by the (random) equation coefficients through the CFL stability condition. However, in general, tasks cannot be ordered by execution time beforehand and therefore a 2-approximation bound is the best that can be generally expected. It is worth emphasizing that the actual makespan is much closer to the optimal in practice, see \Sec{examples}.

Let us now move to the \rev{parallel} case in which a task may require more than one processor. There are several factors to consider in the decision of how many processors are assigned to each task. In a general case \cite{Gmeiner2016}, it depends on two factors, the idling that can be generated and the loss of efficiency in the strong scaling of the underlying solver. If no idling occurs, as in the algorithm we propose herein, and there are enough samples to assign work to all processors, the inevitable loss of efficiency in strong scaling of the solver implies that the best efficiency is always obtained assigning the minimum number of processors that permit to execute each sample, as already noticed in \cite{Eldred1998,Eldred2000}. On the other hand, executing all tasks with the same number of processors would lead to larger execution times on finer levels of the \ac{mlmc} hierarchy. Assuming a weakly scalable solver, and assigning a number of tasks proportional to the size of the mesh at each level gives constant mean execution time for each sample, with variations with the change of the random \ac{pde} coefficients. This assumption is made in \cite{Gmeiner2016}, where the number of processors per task at level $\lev$ of the \ac{mlmc} algorithm is taken of the form $2^{3\lev} p_0$ and $p_0$ is determined to maximize efficiency.

We divide the tasks in charge of each sample in the \ac{mlmc} algorithm into sets $T^q = \{T^q_1,...,T^q_{n_q}\}$ that require $q$ processors with $q \le p$ for fixed $p$ and we denote by $w^q_1,...,w^q_{n_q}$ the execution (wall-clock) time of each of them. The problem of scheduling these tasks using $p$ processors to minimize the makespan is also NP-complete. Therefore, we cannot expect to obtain a polynomial time algorithm, so we propose an algorithm to solve a restricted problem rather than the general one. We restrict the number of possible values of $q$ to $q_0,q_1,\ldots,q_M$ where $q_0|q_1|...|q_M$ (hereafter $a|b$ means that $a$ is an integer divisor of $b$ with both $a,b\in\mathbb{N}$, as usual). Observe that the choice in \cite{Gmeiner2016} satisfies this assumption. We also assume that $q_M|p$. This assumption of divisibility is not necessary in the scheduling algorithm, but permits to prove \Th{2ap} presented below, see \Rm{indep}. The total work is now
\begin{equation}
  \label{eq:total_work_parallel}
  W_{\rm T}=\sum_{q=q_0,q_1,\ldots,q_M} q \sum_{j=1}^{n_q} w^q_j.
\end{equation}
and
$$
W_{\rm LB}=\max \left\{ \overline{W},\max_{q,j}\{w^q_j\} \right\}.
$$
where, again $\overline{W}=W_{\rm T}/p$.

Under these restrictions (of the possible values of $q_0,q_1,\ldots,q_M$) we propose \Alg{incremental_greedy}, whose execution is illustrated in \Fig{greedy-execution} and we prove that it is a 2-approximation in \Th{2ap}.  Observe that because wall-clock times of samples are unknown before their actual execution, implementing \Alg{incremental_greedy} requires dynamic programming \cite{Blazewicz2019}. A \ac{mpi} implementation is described in the next section.

\begin{algorithm}
  \caption{ Incremental greedy scheduling \label{alg:incremental_greedy}}
  \For{$q \in \{q_M,q_{M-1},\ldots,q_1,q_0\}$}{
    \For{$j=1,\ldots,n_q$}{
        Assign $T^q_j$ to a set of $q$ processors among those with less workload. \\
    }
  }
\end{algorithm}

\tikzstyle{every picture}+=[remember picture] 
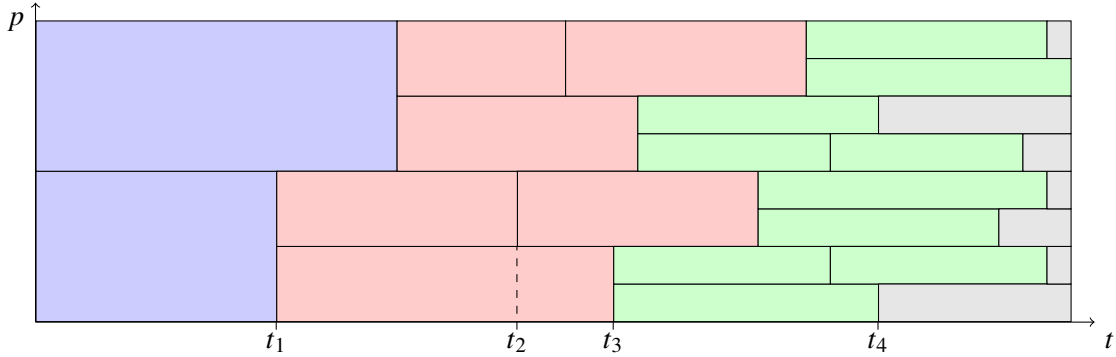
\begin{figure}
\begin{center}
\begin{tikzpicture}[xscale=1.6,yscale=0.5]
\pgfmathsetmacro{\u}{1.6cm}
\pgfmathsetmacro{\v}{0.5cm}
\pgfmathsetmacro{\Lx}{8.8}
\pgfmathsetmacro{\Ly}{8.5}

\draw[->] [black] (0,0) -- (\Lx,0);
\draw[->] [black] (0,0) -- (0,\Ly);
\node[anchor=north west] at (\Lx,0) {$t$};
\node[anchor=north east] at (0,\Ly) {$p$};

\node[anchor=south west,draw,minimum width=2*\u,minimum height=4*\v,fill=blue!20] at (0,0) {};
\node[anchor=south west,draw,minimum width=3*\u,minimum height=4*\v,fill=blue!20] at (0,4) {};

\draw[-] [black] (2,-0.2) -- (2,0);
\node[anchor=north] at (2,0) {$t_1$};

\node[anchor=south west,draw,minimum width=2.8*\u,minimum height=2*\v,fill=red!20] at (2,0) {};
\node[anchor=south west,draw,minimum width=2*\u  ,minimum height=2*\v,fill=red!20] at (2,2) {};
\node[anchor=south west,draw,minimum width=2*\u  ,minimum height=2*\v,fill=red!20] at (3,4) {};
\node[anchor=south west,draw,minimum width=1.4*\u,minimum height=2*\v,fill=red!20] at (3,6) {};

\node[anchor=south west,draw,minimum width=2*\u  ,minimum height=2*\v,fill=red!20] at (4,2) {};
\node[anchor=south west,draw,minimum width=2*\u  ,minimum height=2*\v,fill=red!20] at (4.4,6) {};

\draw[dashed] [black] (4,-0.2) -- (4,2);
\node[anchor=north] at (4,0) {$t_2$};
\draw[dashed] [black] (4.8,-0.2) -- (4.8,0);
\node[anchor=north] at (4.8,0) {$t_3$};

\node[anchor=south west,draw,minimum width=2.2*\u,minimum height=\v,fill=green!20] at (4.8,0) {};
\node[anchor=south west,draw,minimum width=1.8*\u,minimum height=\v,fill=green!20] at (4.8,1) {};
\node[anchor=south west,draw,minimum width=2*\u,minimum height=\v,fill=green!20] at (6,2) {};
\node[anchor=south west,draw,minimum width=2.4*\u,minimum height=\v,fill=green!20] at (6,3) {};
\node[anchor=south west,draw,minimum width=1.6*\u,minimum height=\v,fill=green!20] at (5,4) {};
\node[anchor=south west,draw,minimum width=2*\u,minimum height=\v,fill=green!20] at (5,5) {};
\node[anchor=south west,draw,minimum width=2.2*\u,minimum height=\v,fill=green!20] at (6.4,6) {};
\node[anchor=south west,draw,minimum width=2*\u,minimum height=\v,fill=green!20] at (6.4,7) {};

\node[anchor=south west,draw,minimum width=1.8*\u,minimum height=\v,fill=green!20] at (6.6,1) {};
\node[anchor=south west,draw,minimum width=1.6*\u,minimum height=\v,fill=green!20] at (6.6,4) {};

\node[anchor=south west,draw,minimum width=1.6*\u,minimum height=\v,fill=gray!20] at (7,0) {};
\node[anchor=south west,draw,minimum width=0.2*\u,minimum height=\v,fill=gray!20] at (8.4,1) {};
\node[anchor=south west,draw,minimum width=0.6*\u,minimum height=\v,fill=gray!20] at (8,2) {};
\node[anchor=south west,draw,minimum width=0.2*\u,minimum height=\v,fill=gray!20] at (8.4,3) {};
\node[anchor=south west,draw,minimum width=0.4*\u,minimum height=\v,fill=gray!20] at (8.2,4) {};
\node[anchor=south west,draw,minimum width=1.6*\u,minimum height=\v,fill=gray!20] at (7,5) {};
\node[anchor=south west,draw,minimum width=0.2*\u,minimum height=\v,fill=gray!20] at (8.4,7) {};

\draw[dashed] [black] (7,-0.2) -- (7,0);
\node[anchor=north] at (7,0) {$t_4$};

\end{tikzpicture}
\end{center}
\caption{An example of timeline execution with \Alg{incremental_greedy} with $M=2$ and $q=2q_2=4q_1=8q_0$ and $n_0=10$ (green), $n_1=4$ (red) and $n_2=2$ (blue) samples. At $t=t_1$, the $n_2$ samples requiring $q_2$ processors have been assigned and samples that require $q_1$ processors start to be assigned to the least loaded processors (those in the bottom half). At $t=t_3$, the $n_1$ samples requiring $q_1$ processors have been assigned and samples that require $q_0$ processors start to be assigned. At $t=t_4$ there are no more samples to assign and some processors become idle. Idling time, occurring only at the end of the calculation, is shown in gray.}
\label{fig:greedy-execution}
\end{figure}

\begin{theorem}\label{th:2ap}
  Given $q_0,q_1,\ldots,q_M$ satisfying $q_0|q_1|...|q_M|p$, Algorithm \ref{alg:incremental_greedy} is a 2-approximation.
\end{theorem}
\begin{proof}
  At the beginning the set of $p$ processors is divided into groups $p/q_M$ of $q_M$ processors. Because $q_{l-1}|q_l$, once tasks $T^{q_l}_j, j=1,\ldots,n_{q_l}$ have been executed, each group of $q_l$ processors can be divided into $q_l/q_{l-1}$ groups of $q_{l-1}$ tasks each, which are ready to execute tasks of the group $T^{q_{l-1}}_j, j=1,\ldots,n_{q_{l-1}}$. The previous statement is true for any $l=M,\ldots,0$ and therefore, no processor idles until there are no more tasks to execute. 
  Once this fact is established, the proof is similar to the case of single processor tasks given in \cite{Graham1966} (see \cite{Vazirani2003} for the simplified proof followed here). Consider the group of processors that finishes last and let $w_{\rm s}$ be the starting time of the last executed task and $w_{\rm last}$ its execution time. Because the algorithm assigns tasks to processors with less workload, all the rest of the processors are busy at time $w_{\rm s}$ and therefore $w_{\rm s} < \overline{W}$. Because $w_{\rm last} < \max_{i,j}\{w^i_j\}$ we get $w_{\rm s} + w_{\rm last} < 2W_{\rm LB}<2W_{\rm opt}$.
\end{proof}

\begin{remark}\label{rm:odd}
  In the cases where the hypothesis of \Th{2ap} are not satisfied it is still possible to apply the rationale behind \Th{2ap}. After all tasks $T^{q_l}_j, j=1,\ldots,n_{q_l}$ have been executed,  the group of $q_l$ processors can be divided into $\lfloor q_l/q_{l-1} \rfloor $  subgroups to execute tasks $T^{q_{l-1}}_j, j=1,\ldots,n_{q_{l-1}}$ and there will be $r_l = q_l \mod q_{l-1}$ processors left. If $q_m | r_l $ for some $m<l$ they can be assigned to execute tasks $T^{q_m}$, otherwise they became idle. However, even if these $r_l$ processors can be assigned further work, it is possible that they idle before the rest of tasks have been executed. For instance, consider the case of $q_0=1, q_1=4, q_2=10$. After execution of tasks $T^{10}$ is finished two groups of $4$ processors are created and tasks in $T^4$ start being executed while the remaining two processors execute tasks in $T^1$. If all tasks in $T^1$ are executed before those in $T^4$ these two processors start idling. However, if $n_1 \gg n_4$, which is the case of the \ac{mlmc}, a good balance can be expected in practice.
\end{remark}

\begin{remark}\label{rm:indep}
It is important to keep in mind the hypothesis of independence of tasks made at the beginning of this section. \Th{2ap} is only valid under this assumption. However, this is actually the case in the standard \ac{mlmc} algorithm described in \Sec{smlmc}, apart from some simple post-processing at the end of the execution. In the case of the \ac{amlmc} algorithm described in \Sec{amlmc} additional synchronization occurs at each iteration of the algorithm, introducing a dependency between tasks, i.e. those in the second iteration {\it must} be executed after those in the first. Therefore, the results of this section actually apply to the scheduling of tasks required to complete one iteration of the \ac{amlmc} which, in any case, represents a substantial amount of computational work.
\end{remark}

\section{A message passing implementation}
\label{sec:implementation}

In this section we describe an implementation of the algorithm described in \Sec{scheduling} developed on top of the \ac{mpi} standard. The two main ingredients of the implementation are a strategy for the partition of the set of processors into groups for parallel sampling, described in \Sec{partition}, and a dynamic scheduling algorithm, described in \Sec{dynamic}. We also include a modification that parallelizes the master coordinator that is designed for improving scalability at the most extreme scales, described in \Sec{managers}.

\subsection{Processor partition strategy}
\label{sec:partition}

The whole strategy is based on a master-slave approach, with a master processor, the coordinator, in charge of the decisions required for scheduling \ac{mlmc} tasks and the remaining slave processors in charge of sampling (also referred as slaves in the following). The communication between the master and slaves required to implement \Alg{incremental_greedy} is described in \Sec{dynamic}. We therefore assume we are given $p+1$ processors, one master and $p$ slaves.

As described in \Sec{scheduling}, slave processors will be required to work on different samples and, more importantly, coordinate with different slaves depending on the sample they are executing. Therefore, different partitions are required at different instants during the calculation. To satisfy this requirement we generate a family of partitions of the set of slaves. The total number of partitions, denoted by $M$ as in the previous section, is fixed during the calculation and levels of the \ac{mlmc} algorithm are mapped to each partition. Therefore, the same index $\lev$ is used in the following to denote the \ac{mlmc} level and the partition in which samples of $Y_\lev$ run. When the \ac{amlmc} algorithm described in \Sec{amlmc} is used, a value of $M$ must be defined by the user such that $M>L$ during the whole calculation and the number of samples for levels $\lev>L$ are set to $0$.

The family of partitions is generated in a hierarchical manner, which requires $q_0<q_1< \ldots <q_M<p$. The set of slave processors is first divided into $\lfloor p/q_M \rfloor$  groups of $q_M$ and a group with $r_M=p \mod q_M$ processors. In a second step, each group of $q_M$ processors is divided into $\lfloor q_M|q_{M-1} \rfloor$ groups of $q_{M-1}$ and one group of $r_{M-1}=q_M \mod q_{M-1}$ processors. The remaining group of $r_M$ processors is divided into $\lfloor r_M|q_{M-1} \rfloor$ (which can be $0$) groups of $q_{M-1}$ processors and one group of size $r_{M-1}=r_M \mod q_{M-1}$. The process continues recursively until the final groups of $q_0$ processes are generated. The first processor of each group, which we will refer to as the root, will play an important role in the algorithm described in \Sec{dynamic}.

To fix ideas, consider $p+1=33$ processors that will be used to sample with tasks requiring $q_2=16$, $q_1=8$ and $q_0=4$ processors each, as illustrated in \Fig{good-partitions}. The $p=32$ slave processors will be first split into $2$ groups of $16$ processors each. The second partition of the family will be constructed splitting each of them into $2$ groups of size $8$. Finally these $4$ groups will be split into $2$ sub-communicators having $4$ processors each. Each slave processor will belong to $3$ different groups. In this way, any slave processor may be used for the execution of a task requiring $4$, $16$ or $32$ processors.

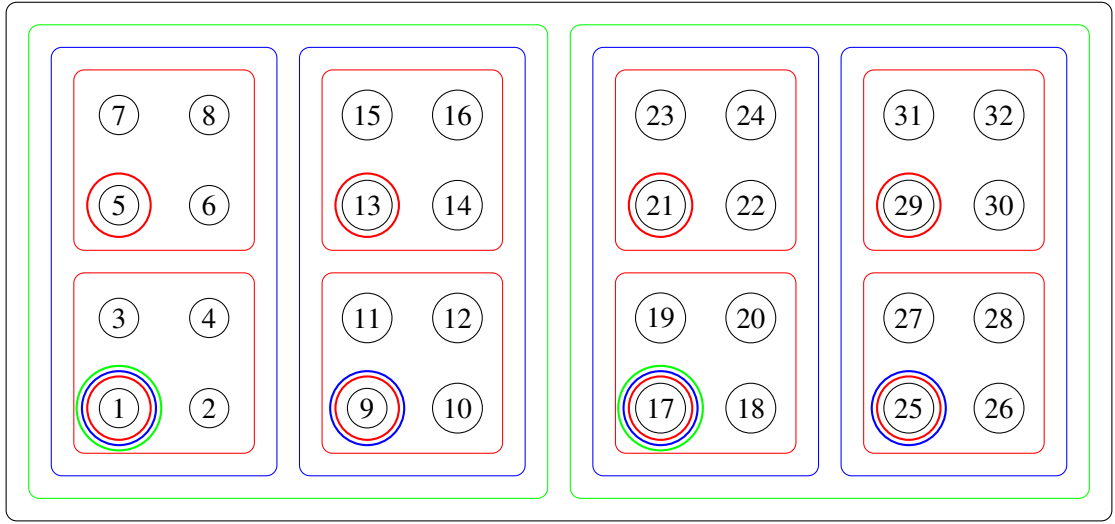
\begin{figure}
\begin{center}
\begin{tikzpicture}[scale=0.3]
\tikzstyle{proc}=[circle,draw=black,minimum size=8pt,inner sep=2pt]
\tikzstyle{root1}=[circle,draw=green,thick,minimum size=32pt,inner sep=2pt]
\tikzstyle{root2}=[circle,draw=blue,thick,minimum size=28pt,inner sep=2pt]
\tikzstyle{root3}=[circle,draw=red,thick,minimum size=24pt,inner sep=2pt]

\draw [rounded corners, black] (0,0) rectangle (49,23);

\newcounter{nodenumber}
\foreach \tx in {0,24}{
  \node [root1] at (\tx+5,5) {};
  \draw [rounded corners, green]  (\tx+1,1) rectangle (\tx+24,22);
  \foreach \sx in {0,11}{
    \node [root2] at (\tx+\sx+5,5) {};
    \draw [rounded corners, blue]  (\tx+\sx+2,2) rectangle (\tx+\sx+12,21);
    \foreach \x in {3}{
      \foreach \y in {3,12}{
        \node [root3] at (\tx+\sx+\x+2,\y+2) {};
        \draw [rounded corners, red]  (\tx+\sx+\x,\y) rectangle (\tx+\sx+\x+8,\y+8);
        \foreach \dy in {2,6}{
          \foreach \dx in {2,6}{ 
            \stepcounter{nodenumber}
            \node [proc] at (\tx+\sx+\x+\dx,\y+\dy) {\arabic{nodenumber}};
          }
        }
      }
    }
  }
}
\end{tikzpicture}
\end{center}
\caption{A set of $33$ processors and their family of partitions into groups of $q_2=16$, $q_1=8$ and $q_0=4$ processors each. Here rank $0$ is the master (not shown) and ranks $1$ to $32$ are slaves. They are first (concurrently) split into groups of $16$ processors represented by green lines with roots $1$ and $17$. These groups are then split into groups represented by blue lines, whose roots are $1$, $9$, $17$ and $25$. Finally the groups are subsequently split into groups having $4$ ranks each, which are represented in red (also used to signal their roots).}
\label{fig:good-partitions}
\end{figure}

On the other hand, consider $p+1=31$ processors that will be used to sample with tasks requiring $q_2=15$, $q_1=6$ and $q_0=3$ processors each, as illustrated in \Fig{bad-partitions}. The $p=30$ slave processors will be first split into $2$ groups of $15$ processors each. The second partition of the family will be constructed splitting each group into $2$ groups having $6$ processors and one group having $3$ processors. Finally, each group of the second partition will be split into one or two groups having $3$ processors.

\begin{figure}
\begin{center}
\begin{tikzpicture}[scale=0.3]
\tikzstyle{proc}=[circle,draw=black,minimum size=8pt,inner sep=2pt]
\tikzstyle{idle}=[rectangle,draw=black,minimum size=8pt,inner sep=6pt]
\tikzstyle{root1}=[circle,draw=green,thick,minimum size=32pt,inner sep=2pt]
\tikzstyle{root2}=[circle,draw=blue,thick,minimum size=28pt,inner sep=2pt]
\tikzstyle{root3}=[circle,draw=red,thick,minimum size=24pt,inner sep=2pt]

\draw [rounded corners, black] (0,0) rectangle (35,34);

\newcounter{nodenum}
\foreach \sx in {0,17}{
  \foreach \sy in {0}{
    \node [root1] at (\sx+5,5) {};
    \draw [rounded corners, green]  (\sx+1,1) rectangle (\sx+17,33);

    \node [root2] at (\sx+5,5) {};
    \draw [rounded corners, blue]  (\sx+2,2) rectangle (\sx+16,13);
    \node [root2] at (\sx+5,17) {};
    \draw [rounded corners, blue]  (\sx+2,14) rectangle (\sx+16,25);
    \node [root2] at (\sx+5,29) {};
    \draw [rounded corners, blue]  (\sx+2,26) rectangle (\sx+16,32);

    \foreach \dy in {3,8,15,20,27}{
      \node [root3] at (\sx+5,\dy+2) {};
      \draw [rounded corners, red]  (\sx+3,\dy) rectangle (\sx+15,\dy+4);
      
      \foreach \dx in {5,9,13}{ 
        \stepcounter{nodenum}
        \node [proc] at (\sx+\dx,\dy+2) {\arabic{nodenum}};
      }
    }
  }
}

\end{tikzpicture}
\end{center}
\caption{A set of $31$ processors and their family of partitions into groups of $q_2=15$, $q_1=6$ and $q_0=3$ processors each. Here rank $0$ is the master (not shown) and ranks $1$ to $30$ are slaves. They are first (concurrently) split into two groups of $15$ processors represented by green lines with roots $1$ and $16$. These groups are then split into groups represented by blue lines, whose roots are $1$, $7$, $13$, $16$, $22$ and $28$. Finally these groups are subsequently split into groups having $3$ ranks each, which are represented in red (also used to signal their roots).}
\label{fig:bad-partitions}
\end{figure}
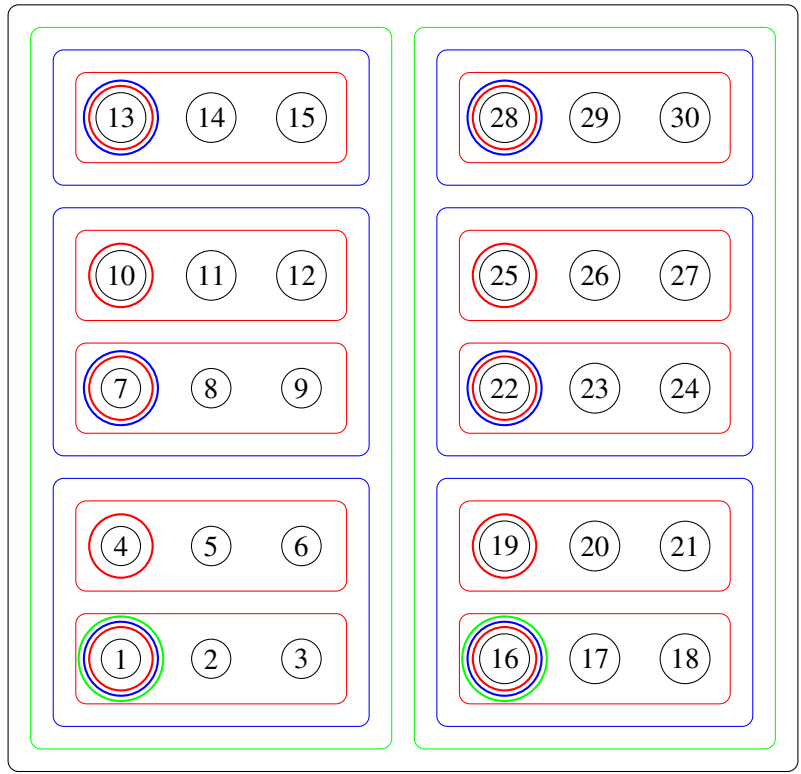

The groups of processors are managed using the primitives provided by the \ac{mpi} standard acting on \ac{mpi} communicators. This allows copying \vbt{MPI\_comm\_world} using \vbt{MPI\_comm\_dup} and then to split it into new (sub-)communicators calling \vbt{MPI\_comm\_split} (a collective operation), which requires a color (defined by an integer) to identify the group. In the first partition the color is given by $t/q_M$ where $t$ is the identifier of the current task in the global communicator (a copy of \vbt{MPI\_comm\_world}), which can be obtained calling \vbt{MPI\_comm\_rank}. The second splitting is performed in the same fashion, now with a color defined by $t/q_{M-1}$ and $t$ the current task in the first communicator just created. This procedure is easily coded in a recursive function with the current communicator as argument.

\subsection{Dynamic scheduling algorithm}
\label{sec:dynamic}

In this section, we describe the proposed dynamic scheduling in \Alg{dyn_par_sched}. At the beginning of the execution, the family of partitions described in the previous section is created, which is a collective operation. After this initial phase, which also includes the allocation of required data structures, the calculation of the \ac{mlmc} estimator proceeds as shown in \Alg{dyn_par_sched} (see \Line{SumYs}, \Line{ReduceYs}, and \Line{SumMeanYs}).
 
Although the \ac{mpi} standard permits a multiple-program multiple-data execution model, the \ac{spmd} execution model is widely used and the implementation described here is made in this way. In this programming model, it is possible to perform different tasks on different processors coding the logic in terms of the \ac{mpi} ranks, which are obtained calling \vbt{MPI\_comm\_rank}, as mentioned before. A call to this function on the global communicator permits determining if the current process is the master (with rank $0$) or a slave (with positive rank) and a call on a given subcommunicator if it is the root (with rank $0$) of that group or not. This functionality is implemented in functions \vbt{i\_am\_master}, \vbt{i\_am\_slave}, and \vbt{i\_am\_root} (which receive the current level to specify the partition) that return true or false. It is also worth emphasizing that in the \ac{spmd} programming model variables initialized in \Line{init1} to \Line{init3} of \Alg{dyn_par_sched} are available in all processors, possibly with different (or no) use on each of them.

The communications are organized in two layers: between the master and the roots and between the roots and the rest of slaves. The master does not exchange any message, and therefore does not require any synchronization, with a slave that is not a root on some partition. Roots are in charge of communicating with the master to receive duties and with the rest of slaves to send them. There are two duties, \vbt{DoSample} and \vbt{NextPartition}, that the master assigns to each root (and roots to the rest of slaves).\footnote{Actually it is possible to have more duties, such as \vbt{allocate} and \vbt{free}, but their need depends on the software design and are not relevant for the description of the algorithm.}

All slave processors start working in the finest partition $\lev=M$ (see the initialization at \Line{init3} in \Alg{dyn_par_sched}) and emit a broadcast to get the duties from their root in \Line{Broadcast} of \Alg{dyn_par_sched}. This broadcast is executed in the slave communicators of level $\lev$, e.g. in the green communicators shown in \Fig{good-partitions} and all slaves wait for their root (processors $1$ and $17$ in \Fig{good-partitions}) to get the duty. Each root communicates to the master that the group is ready to execute at level $\lev$ (see \Line{SendToMaster}) and receives the instruction to either sample (with the duty \vbt{DoSample}) or proceed to the next partition (with the duty \vbt{NextPartition}) at \Line{ReceiveFromMaster}. This duty is defined by the master at \Line{DefineDuty} depending on the number of samples required at that level, which is received from the root at \Line{ReceiveFromRoot}. The function \vbt{ReceiveFromRoot} is implemented calling \vbt{MPI\_Recv} with the wildcard \vbt{MPI\_ANY\_SOURCE} for the argument \vbt{source}, which permits receiving a message from any of the roots. While there are samples to assign, the master keeps looping at \Line{ContinueSampling} actively listening from roots. Following the example of \Fig{good-partitions}, roots $1$ and $17$ send their rank (\vbt{RootId} in \Alg{dyn_par_sched}) and the level $\lev=2$ (green) to the master, who receives one of them and assigns a sample at \Line{SendDoSample} at each iteration of the loop in \Line{ContinueSampling}. In this way the execution timeline shown in \Fig{greedy-execution}, which is consistent with the partition set partition shown in \Fig{good-partitions}, starts with the $n_2=2$ samples assigned. When the roots receive the duty \vbt{DoSample} and broadcast it, all the slaves in the group start sampling at \Line{Sample}.

While all the slaves are sampling, the master is waiting at \Line{ReceiveFromRoot} until a group finishes, which occurs at time $t_1$ in \Fig{greedy-execution}. When slaves finish sampling they go to the next iteration of the loop in \Line{ContinueSampling} and arrive again to the broadcast at \Line{Broadcast}, except the root that sends its rank and partition to the master at \Line{SendToMaster}. After receiving this message at \Line{ReceiveFromRoot}, the master proceeds then to decide again if there are samples to assign at \Line{DefineDuty}. In the execution timeline of \Fig{greedy-execution} all the $n_2$ samples have been assigned at time $t_1$ so the master sends the duty \vbt{NextPartition} to root $1$ at line \Line{SendNextPartition}. After receiving and broadcasting this duty, slaves in this group go to the next partition at \Line{NextPartition} and the group is divided into two groups shown in blue in \Fig{good-partitions} with roots $1$ and $9$. In the next iteration of the loop, they arrive again to \Line{Broadcast} and to \Line{SendToMaster} but now with $\lev=1$. After the communication with the master they receive the duty \vbt{DoSample} and they start sampling at level $\lev=1$ while processors $17$ to $32$ are sampling at level $\lev=2$. The same process occurs at time $t_3$, now reaching level $\lev=0$. When slaves finish sampling, if there are still samples to assign, the master assigns them samples at the same level, which occurs at time $t_2$ in \Fig{greedy-execution}. Observe that it may even happen that samples assigned on partitions separated by more than one level are executed concurrently.

At some point all samples have been assigned and $n_\lev=N_\lev$ for all $\lev=0,\ldots,M$. When the next sample finishes, which occurs at time $t_4$ in \Fig{greedy-execution}, the master will send it down in the partition hierarchy and if $\lev=0$ that will imply for the slaves to go to $\lev=-1$ which will make them to break the loop at \Line{BreakContinueSamplingSlaves}. The master will count the roots that left the execution and will break the loop at \Line{BreakContinueSamplingMaster} when they reach $p/q_0$ ($8$ in the execution timeline of \Fig{greedy-execution}). 

To reduce communication, the samples are not typically communicated one at a time, but in batches of samples. These batches are always sequential ranges of integers which partition the set $\{1,\cdots,N_\lev\}$. Instead of communicating each individual sample, instead each individual batch of samples is communicated, by indicating the smallest and largest sample in the batch. This batch size, denoted $b$ is determined dynamically, depending on the current number of assigned samples on the particular level, $n_\lev$, and the total number of assigned samples on the particular level, $N_\lev$, as well as the number of partitions on that level, $p|q_\lev$. Each batch size proportional to the current number of unassigned samples, denoted by $N_\lev - n_\lev$, divided by the total number of samples $N_\lev$. More specifically,
\begin{equation}
\label{eq:batch}
 b_\lev = \lceil\frac{(N_\lev-n_\lev)}{N_\lev(p|q_\lev)}\rceil,
\end{equation}
subject to minimum and maximum batch sizes that can be modified by parameters. The default minimum is \rev{1\% of } $\lceil N_\lev/(p|q_\lev)\rceil$ (the total number of samples per partition), and the default maximum is \rev{$\varphi^{-1}$} of $\lceil N_\lev/(p|q_\lev)\rceil$, \rev{where $\varphi$ is the golden ratio (so the ratio between the total and unassigned samples equals the ratio between unassigned and assigned ones), $\varphi\approx 0.62$} . These are the values used in all examples here. Naturally, should the size of the batch be larger than the number of remaining samples, the batch only contains the remaining samples.

After the sampling loop at \Line{ContinueSampling} is finished, there is a reduction process to get the final estimate $\widetilde{Q}_L$. Observe that partial sums have been performed during sampling, keeping them on root processors, see \Line{SumYs}. At the end of the sampling loop, slave processors send these partial sums to the master at \Line{SendYsToMaster}, which are received (using local variables $Y_\lev$) at \Line{ReceiveYsFromRoot}, accumulated at \Line{ReduceYs} and finally weighted and added to compute the final estimate in \Line{SumMeanYs}. If executing \ac{amlmc}, the number of samples and the maximum level are updated in \Line{UpdateL} and \Line{UpdateN} respectively and the whole process is repeated until convergence is reached.

\begin{algorithm}
  \caption{Dynamic parallel scheduling to obtain  $\widetilde{Q}_L$ := \ac{amlmc}($L_0$,$M$,$\{N^0_\lev\}$,$\epsilon$)\label{alg:dyn_par_sched}}
  Set $N_\lev=N^0_\lev $, $n_\lev=0$, $\overline{Y_{\lev}} = 0$, for all $\lev=0,\ldots,L_0$ \; \nllabel{line:init1} 
  Set $N_\lev=0 $, $n_\lev=0$, $\overline{Y_{\lev}} = 0$, for all $\lev=L_0+1,\ldots,M$ \;\nllabel{line:init2}
  Set $L=L_0$, $\lev=M$, $Q_L$, $\widetilde{Q}_L = 0 $, $nr=0$, $e=\epsilon$ \;\nllabel{line:init3}
  \While{$e>\epsilon$ and $L<M$}{
    ContinueSampling=true\;
    \While{ContinueSampling}{ \nllabel{line:ContinueSampling}
      \uIf{i\_am\_master}{
        ReceiveFromRoot(RootId,$\lev$)\; \nllabel{line:ReceiveFromRoot}
        \eIf{$n_\lev<N_\lev$}{ \nllabel{line:DefineDuty}
          Set $n_\lev=n_\lev+ b_\lev(n_\lev, N_\lev, p|q_\lev)$ \rev{with $b_\lev$ given by \Eq{batch} or \Eq{batch_par}} \;
          SendToRoot(RootId,DoSample)\; \nllabel{line:SendDoSample}
        }{
          SendToRoot(RootId,NextPartition)\; \nllabel{line:SendNextPartition}
          \If(\tcc*[h]{Count roots that finished}){$\lev=0$}{
            $nr=nr+1$ \;
            \If{$nr=p/q_0$}{
              Set ContinueSampling=false\; \nllabel{line:BreakContinueSamplingMaster}
            }
          }
        }
      }
      \ElseIf{i\_am\_slave}{
        \If{i\_am\_root($\lev$)}{
          SendToMaster(RootId,$\lev$)\; \nllabel{line:SendToMaster}
          ReceiveFromMaster(duty)\; \nllabel{line:ReceiveFromMaster}
        }
        Bcast(duty)\; \nllabel{line:Broadcast}
        \uIf{duty is sample}{ \nllabel{line:Sample}
          Compute $Y_\lev$\;
          \If{i\_am\_root($\lev$)}{
            Set $\overline{Y_{\lev}} = \overline{Y_{\lev}} + Y_\lev$ \; \nllabel{line:SumYs}
          }
        }
        \ElseIf{duty is NextPartition}{ \nllabel{line:NextPartition}
          Set $\lev=\lev-1$.\;
          \If{$\lev < 0$}{
            Set ContinueSampling=false\; \nllabel{line:BreakContinueSamplingSlaves}
          }
        }
      }
    }
    \uIf{i\_am\_master}{
      \For{$i=1$ \KwTo $p/q_0$}{
        ReceiveFromRoot(Rootid,$\{Y_{\lev}\}_{\lev=0}^L$)\; \nllabel{line:ReceiveYsFromRoot}
        Set $\overline{Y_{\lev}}=\overline{Y_{\lev}}+Y_\lev$ for $\lev=0,\ldots,L$ \; \nllabel{line:ReduceYs}
      }
      Set $\widetilde{Q}_L = \sum_{\lev=0}^L \overline{Y_{\lev}}/N_\lev$\; \nllabel{line:SumMeanYs}
      \rev{Compute $e$ from \Eq{mlmc_error} and variance and cost estimates}\;
      Update $L$ using \Eq{mlmc_num_levels} with $c$ and $\alpha$ estimated from $\overline{Y_{\lev}}$\; \nllabel{line:UpdateL}
      Update $N_\lev $ for all $\lev=0,\ldots,L$ using \Eq{mlmc_num_samples} \; \nllabel{line:UpdateN}
    }
    \ElseIf{i\_am\_slave}{
      \If{i\_am\_root($0$)}{
        SendToMaster($\{\overline{Y_{\lev}}\}_{\lev=0}^L$)\; \nllabel{line:SendYsToMaster}
      }
    }
  }
      
\end{algorithm}

\subsection{Parallelizing the coordinator}
\label{sec:managers}
The above implementation has a notable potential bottleneck in the use of a single master task that is seen in the example \Sec{Pause}. This bottleneck is not currently an issue at the scale of computation usually encountered, but a potentially fast executing model requires significantly more communication throughput. It is indicative of a computational state wherein the overall processing power dramatically overshadows that of the evaluation of a single model, which is generally expected to occur at some large enough number of processors, regardless of the model.

It is anticipated that this is an important potential bottleneck in certain applications and on certain architectures. We therefore consider briefly the parallelization of the master coordinator. This is done by replacing the single task by a communication tree. Each communication task is limited to have a specified number of tasks which report to it, be they root tasks, or other coordinator tasks. This parameter is herein referred to as \vbt{comm\_limit}. This functionality works with the batch sampling mechanism mentioned in \Sec{dynamic}. Samples are allocated based not only on total and available samples, but also the number of executing partitions on the slave coordinator $i$, denoted $P_\lev(i)$, with respect to the total number of partitions that report to that master coordinator , $\sum_k P_\lev(k)$. In this case the batch, $b$, is proportional to
\begin{equation}
\label{eq:batch_par}
 b_\lev = \lceil\frac{(N_\lev-n_\lev)P_\lev(i)}{N_\lev\sum_k P_\lev(k)}\rceil,
\end{equation}
and is again subject to similar minimum and maximum values.

Functionally, this implementation can be visualized as additional partition layers of tasks as in \Fig{good-partitions} or \Fig{bad-partitions}, but with an additional task separated on each new layer representing this dedicated communication task. The layers are usually identical, but often, there is a single communication partition that can only be partially filled.
In such a case, that potential communication partition is lumped entirely into another communication partition. This implementation does not balance these loads and allows this limit on communication tasks to be violated. However, this violation occurs only on at most one task at each level of the tree, and in the worse case doubles the amount of communication on that layer. 

The implementation improves scalability for the example of \Sec{Pause}, but uses a relatively crude communication tree, and is not especially optimal. Additionally, all manager tasks exist outside of the context of the largest computational partition, which means that there is a lower bound on the roots-per-task value. For example, in an example where $L=2$ and level $0$ uses $2$ task per model evaluation, and level $1$ uses $100$ tasks per model evaluation, the lower bound on the roots-per-task value would be $100/2=50$. This is a potential problem if there are both rapidly executing serial models, and slowly executing highly parallelized models being scheduled.

\section{Numerical Examples}
\label{sec:examples}

In this section we evaluate the performance of the scheduling algorithm which was implemented in 
\FEMPAR{} \cite{badia_fempar:_2017,Badia2020}, an open source, hybrid OpenMP/MPI parallel, object-oriented (OO) FORTRAN200X scientific software framework for the massively parallel FE simulation of multiscale/multiphysics problems governed by PDEs. We developed a new layer of software that permits the execution of several instances of a model as required by sampling \ac{uq} methods. In this way it is possible to exploit advanced discretization methods provided by the \FEMPAR{} library, e.g., \ac{agfem} methods. These embedded techniques have been recently exploited in \cite{Badia2021} to perform \ac{uq} in random geometries using the \ac{mlmc} implementation described herein. \rerev{On top of these advanced discretization methods, \FEMPAR{} provides advanced built-in multilevel domain decomposition solvers \cite{badia_multilevel_2016,Badia.etal.ACME.2013} and an interface with PETSc \cite{petsc-web-page}, which, in turn, provides access to many other linear and nonlinear solvers.} The code is open source, under the GPLv3 license, and available at the \FEMPAR{} gitlab repository 
(\href{https://github.com/fempar/fempar}{https://github.com/fempar/fempar}).

\rev{Most numerical experiments presented herein have been performed in Marenostrum IV (MN-IV), a supercomputer hosted at the Barcelona Supercomputing Center (BSC). MN-IV is equipped with 3456 compute nodes connected together with the Intel OPA HPC network. Each node is equipped with 2x Intel Intel Xeon Platinum multi-core CPUs, with 24 cores each (i.e., 48 cores per node), and 96 GBytes of RAM. \FEMPAR{} was compiled with the Intel Fortran compiler (v18.0.1) with the system-recommended optimization flags, and linked against the Intel MPI Library for message-passing (2018.0.128), the BLAS/LAPACK and PARDISO available on the Intel MKL library for optimized dense linear algebra kernels, and sparse direct solvers, respectively}\rerev{, and the PETSc library (3.12.4).}

\rev{Some additional tests were perfomed in Gadi, a supercomputer hosted at the Australian National Computational Infrastructure Agency (NCI). Gadi is a equiped with 3024 nodes, each containing 2x 24-core Intel Xeon Scalable \textit{Cascade Lake} cores and 192 GB of RAM. All nodes are interconnected via Mellanox Technologies' latest generation HDR InfiniBand technology.
\FEMPAR{} was compiled with the GNU Fortran compiler (v8.5.0) with the system-recommended optimization flags, and linked against the OpenMPI Library for message-passing (4.2.2), and the BLAS/LAPACK and PARDISO available on the Intel MKL library (v2020.0.166) for optimized dense linear algebra kernels, and sparse direct solvers, respectively}\rerev{, and the PETSc library (3.12.4).}

We consider two examples aimed at testing the scalability and efficiency of the scheduling software.
First, we consider a model independent evaluation of performance, testing the scheduling algorithm using a computational model that simply idles \rev{for a randomly selected time interval}, while executing in parallel in \Sec{Pause}. 

As a second example, we consider a 3D Poisson problem in \Sec{Popcorn}, with a random diffusion term that is described by a \ac{KLE} and fixed Dirichlet boundary conditions on the boundary of a randomized popcorn geometry similar to that described in \cite{Badia2021}. This Poisson model is discretized using \ac{amg}, and executed in parallel with differing core counts on each level. 

For each example we perform the standard scaling experiments, namely,
\begin{itemize}
\item strong scaling where the number of \ac{mlmc} samples (or \ac{amlmc} tolerance) is kept constant, while the number of processors is increased,
\item weak scaling where the number of \ac{mlmc} samples is increased (or \ac{amlmc} tolerance is decreased), and the number of processors is increased, so that the relative number of \ac{mlmc} samples to processors is fixed.
\end{itemize}

We consider here sample sizes (or initial sample sizes in the case of \ac{amlmc}), $N_\ell$, which are integers that experience an exponential growth,
\begin{equation}
\label{eq:sample_scaling}
N_l \propto \exp\left(\Gamma \log(s)(L-l)\right)
\end{equation}
(for $l=0,...,L$). 
\rerev{As defined in \Sec{smlmc},}\rev{$s$ is the ratio between the size of the mesh at two different levels of the hierarcy, $s=2$ in all the experiments presented herein. The constant $\Gamma$ defines the relation between the number of samples at different levels required to achieve optimal complexity by equilibrating statistical and spatial approximation errors. The optimal value depends on the properties of the discretization and solution algorithms as well as on the regularity of the problem. In the case of the solution of the Poisson problems in smooth domains using linear \acp{fe} and solving the resulting linear system with the conjugate gradient method, $\Gamma=4$ is optimal, as validated in~\cite{Badia2021}. It is worth noting that this is not the case of the popcorn example of Section 5.2. The main motivation for the development of \ac{amlmc} methods is precisely to correct the number of samples $N_\ell$ during the excution of the algorithm to obtain an optimal complexity regardless of the value of $\Gamma$.}

To test the scalability of the scheduling algorithm, we scale the overall computational cost of the experiments by a sample multiplier, denoted by $C$, taking the number of samples at each level according to
\begin{equation}
\label{eq:sample_scaling_constant}
 N_\ell = C N_\ell^{\star},
\end{equation}
where $N_l^{\star}$ is a relatively small sample size obtained (approximately) from \Eq{sample_scaling}, whose values are given in \Tab{example_samples}.

\begin{table}[htbp]
\centering
 \caption{Sample sizes by Level given by \Eq{sample_scaling}.}
 \begin{tabular}{|c|c|c|c|c|}
 \hline
 $N^{\star}_0$ & $N^{\star}_1$ & $N^{\star}_2$\\
 \hline
 $1024$ & $64$ & $4$\\
 \hline
 \end{tabular}
 \label{tab:example_samples}
\end{table}

In the results below we present the speedup $S(p)$ defined by
\begin{eqnarray}
  \label{eq:speed_strong}
  S(p)& =  \frac{t_w(1)}{t_w(p)}.
\end{eqnarray}
where $t_w(p)$ is the wall clock time to compute the \ac{mlmc} estimate as a function of $p$.
The speedup is useful here for weak scaling, as this normalization allows us to compare results for different model evaluation times, which would be difficult to do if wall times were presented.

We also consider the effect of idling and coordination relative to the time spent in active computation. We divide the scheduling algorithm execution on each task into three parts, managing, $m(p)$, idling, $i(p)$, and active computation, $a(p)$, which we measure in core-seconds. The master task and any other coordinators are always contributing to the managing total, while executing tasks are contributing to either idling or active computation, depending on where they are in the scheduling algorithm execution. When an executing partition is in the setup or run phase, its cores contributes to the active total; otherwise, it contributes to the idling total. We define the efficiency as in \cite{Sukys2014}, to be
\begin{eqnarray}
 \label{eq:active_ratio}
 A(p)& =  \frac{a(p)}{pt_w(p)},
\end{eqnarray}
a value in $(0,1)$ that measures the overhead of the scheduling algorithm as a ratio of total core time used in the computation.

We note that our results compare most closely to the similarly performed experiments in \cite{Gmeiner2016} and \cite{Tosi2021}. Our implementation performs well generally, and especially well at addressing the difficult problem of scalability when the model has a low execution time.

\subsection{A benchmarking example}
\label{sec:Pause}
Our first example is a conceptually simple test of the scheduling algorithm scaling with a simple model class which utilizes partitions of individual processors, and whose only evaluation is to wait for a random amount of time. This model allows us to increase the stress from the communication in the scheduling algorithm. For each sample we define this time by a uniform distribution of mean $\mu$ and variance $\sigma$, i.e. a random value between $\mu-\sqrt{3\sigma}$ and  $\mu+\sqrt{3\sigma}$.
This exact knowledge of the statistics of the time required for the model evaluation permits an accurate benchmarking of the scheduling algorithm.

For the parallel models considered here we do not vary the execution time based on the level of computation, which is a reasonable approximation when the size of the problem scales in correspondence with the number of processors used to compute it. This model is quite robust to evaluating the scalability of the scheduling algorithm independently of the specific computations of the model, and with the following examples, demonstrates the benefits to parallelization of the models being used.  We consider $\mu$ in $\{10^{-4}\text{s}, 10^{-2}\text{s}, 1\text{s}\}$. Here, the short computation time of $10^{-4}$ allows us to insure a high sustained level of communication in the scheduling algorithm as the number of tasks increases, stressing the communication structure, while the longer computation times may more accurately simulate real model evaluation times.

Specifically, we consider a benchmarking example that is discretized to run on $q_0=8, q_1=64$, and $q_2=512$ cores. We refer to the $10^{-4}$s average as the Short Time Model, the $10^{-2}$s average as the Medium Time model, and the $1$s average as the Long Time model. Here in particular, this model is designed to represent a 3 dimensional \ac{amg} solver.

We also consider the same model parameters under strong scaling, with the sample multiplier $C$ fixed to $128$, while for weak scaling $C$ is the number of computational nodes. The samples for each level are given by \Eq{sample_scaling_constant} and \Tab{example_samples} for $\{16,32,64,128,256,512,600\}$ nodes.

\Fig{ParPauseScaling} shows the speedup for the two scaling regimes given by \Eq{speed_strong}, while \Fig{ParPauseNCR} shows the efficiency. We note that the parallel benchmarking performs well for the Medium Time and Long Time models, but that the Short Time model provides a significant communication stress, and that in this case, the use of additional communication tasks improve scalability.

We also consider two additional modifications, beyond the consideration of additional communicators and execution times. For these, we use only one communicator task, and $\mu = 10^{-2}$s. First, is the variance of the execution time.  We consider three cases: that of high $\sigma$, given by $\sigma_{H}=0.5\mu/\sqrt{3}\approx 0.289\mu$; that of medium $\sigma$, given by $\sigma_{M}=0.2\mu$; and that of low $\sigma$, given by $\sigma_{L}=0.01\mu$. These results are shown in \Fig{MultiSigmaScaling}. 
We note that changes to $\sigma$ provide at most a marginal change in performance. \rev{Therefore, we can conclude that the variation of execution times among samples does not significantly affect the scalability properties of the algorithm, neither in a weak or strong scaling case.}

Additionally, we consider the use of batch sampling. Using a single master communicator, we consider the case when all batch sizes are of size $1$ for $\mu = 10^{-2}$s. These results are shown in \Fig{NoBatchScaling}. 
It can be seen that there is worse scalability at higher core counts without batch sampling, as the single master communicator cannot act fast enough to efficiently communicate each individual sample to the execution partitions. \rev{Therefore, we can conclude that the use of batching sampling can be useful to improve scalability for extreme core counts.}

\rev{We end this section with the analysis of the influence of the number of cores in the scalability. In \Fig{MultiProcs} we compare the performance
of two different sets of processors for sampling, namely, $q=8,64,512$ and $q=9,81,729$. As it can be observed in \Fig{MultiProcs} the influence of this choice is minimal, 
as anticipated in \Rm{odd}. Although this result could be surprising, we note that there is always a master processor in charge of the scheduling that introduces an oddness
and therefore the set of processors is never partitioned in an optimal way in the architectures considered herein. For example, the leftmost point in all figures corresponts to 16 nodes with 48 cores, that is, $p+1=768$ processors, one of which is used for the master. The partition of $p=767$ will result in some idling in both cases. In fact, in the first case
it is possible to exploit $512+3\times64+7\times8=760$ processors whereas in the second it is possible to exploit $729+4\times9=765$ processsors. In any case, this idling does 
not increase with the number of nodes and scalability is not affected.}

\begin{figure}
  \centering
  \begin{subfigure}[b]{0.45\textwidth}
    \centering
    \includegraphics[scale=0.4]{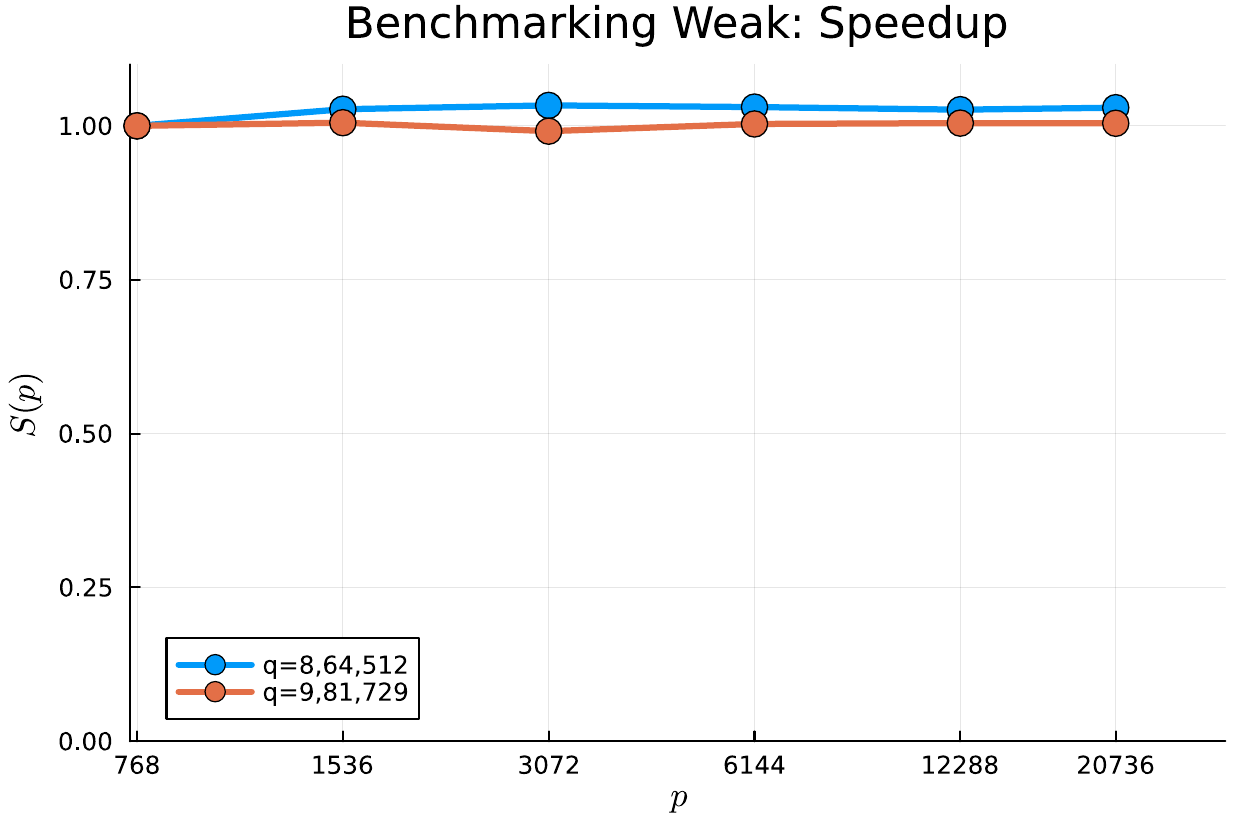}
  \end{subfigure}
  \hfill
  \begin{subfigure}[b]{0.45\textwidth}
    \centering
    \includegraphics[scale=0.4]{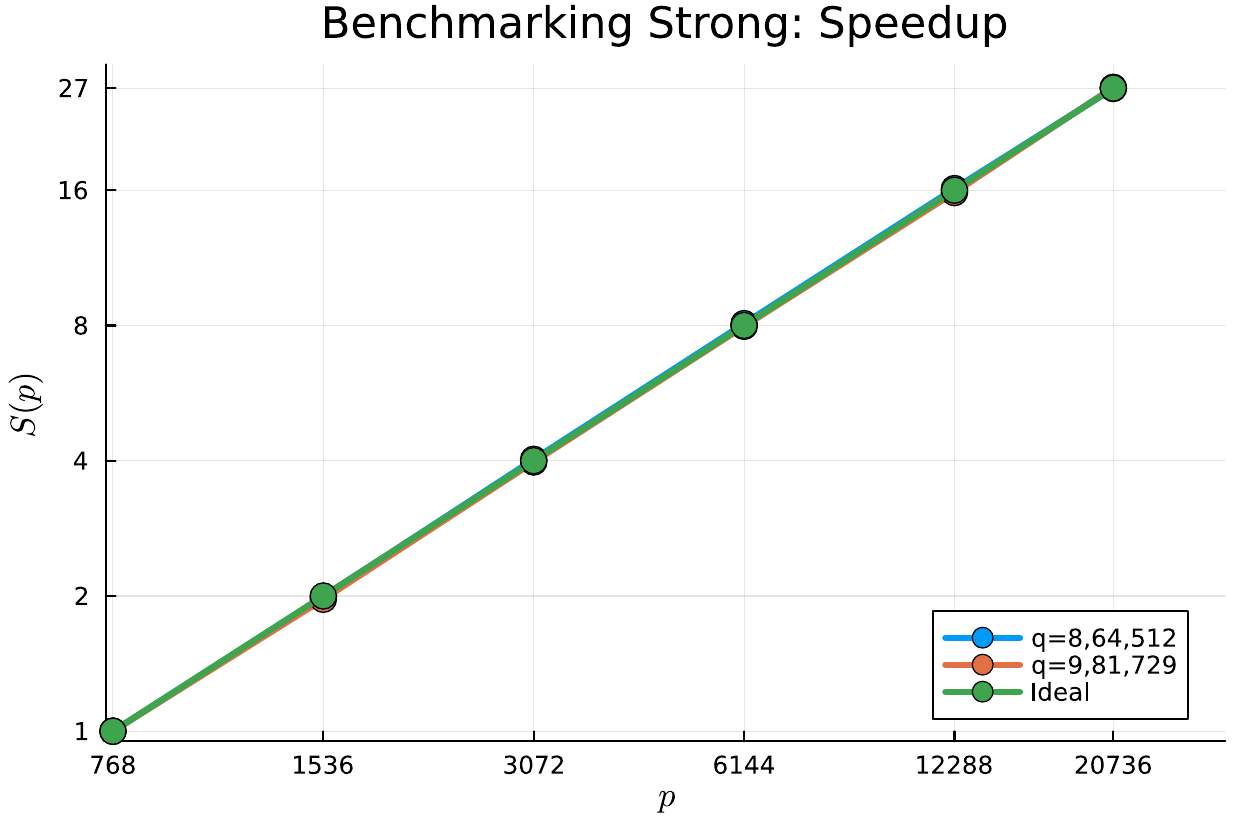}
  \end{subfigure}
  \caption{Benchmarking speedup under weak scaling (left) and strong scaling (right) \rev{for different sampling times (short, medium and large), with and without parallel coordinator of the scheduling}.}
  \label{fig:ParPauseScaling}
\end{figure}

\begin{figure}
  \centering
  \begin{subfigure}[b]{0.45\textwidth}
    \centering
    \includegraphics[scale=0.4]{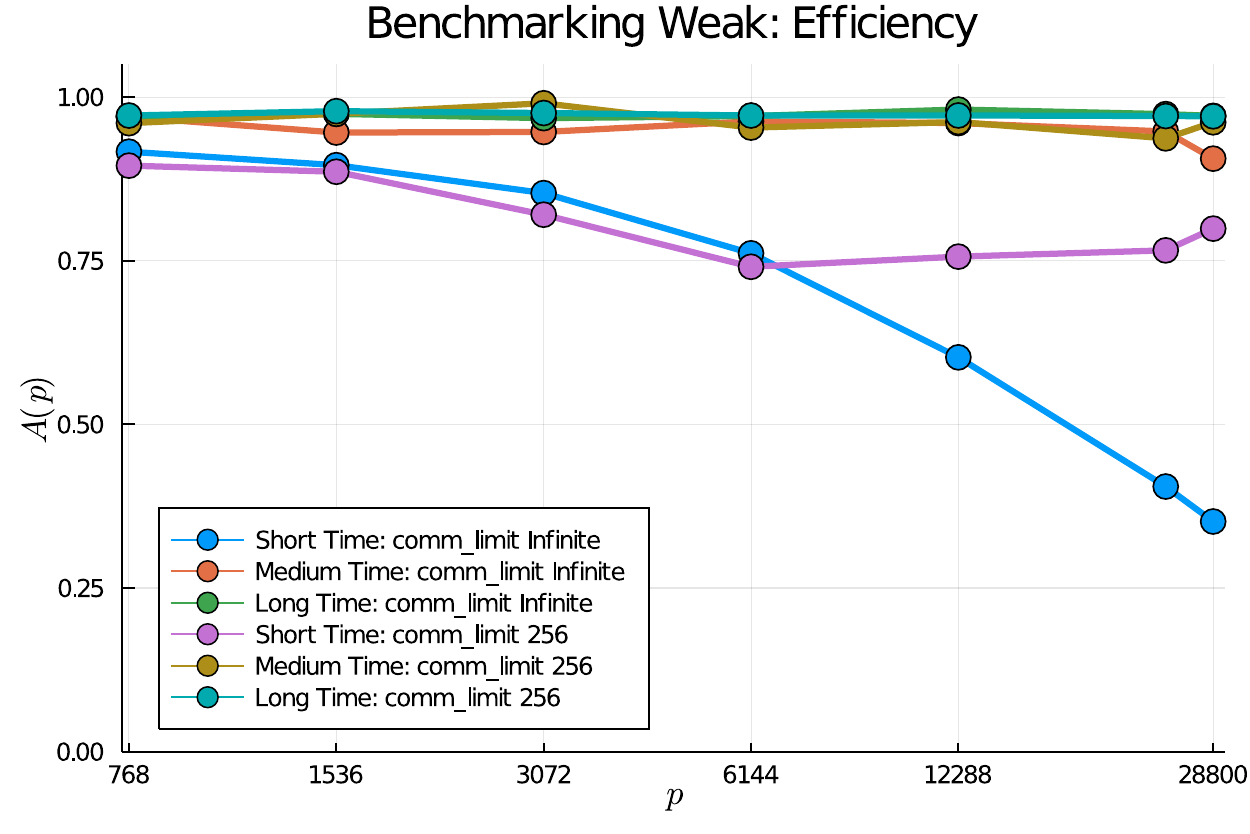}  
  \end{subfigure}
  \hfill
  \begin{subfigure}[b]{0.45\textwidth}
    \centering
    \includegraphics[scale=0.4]{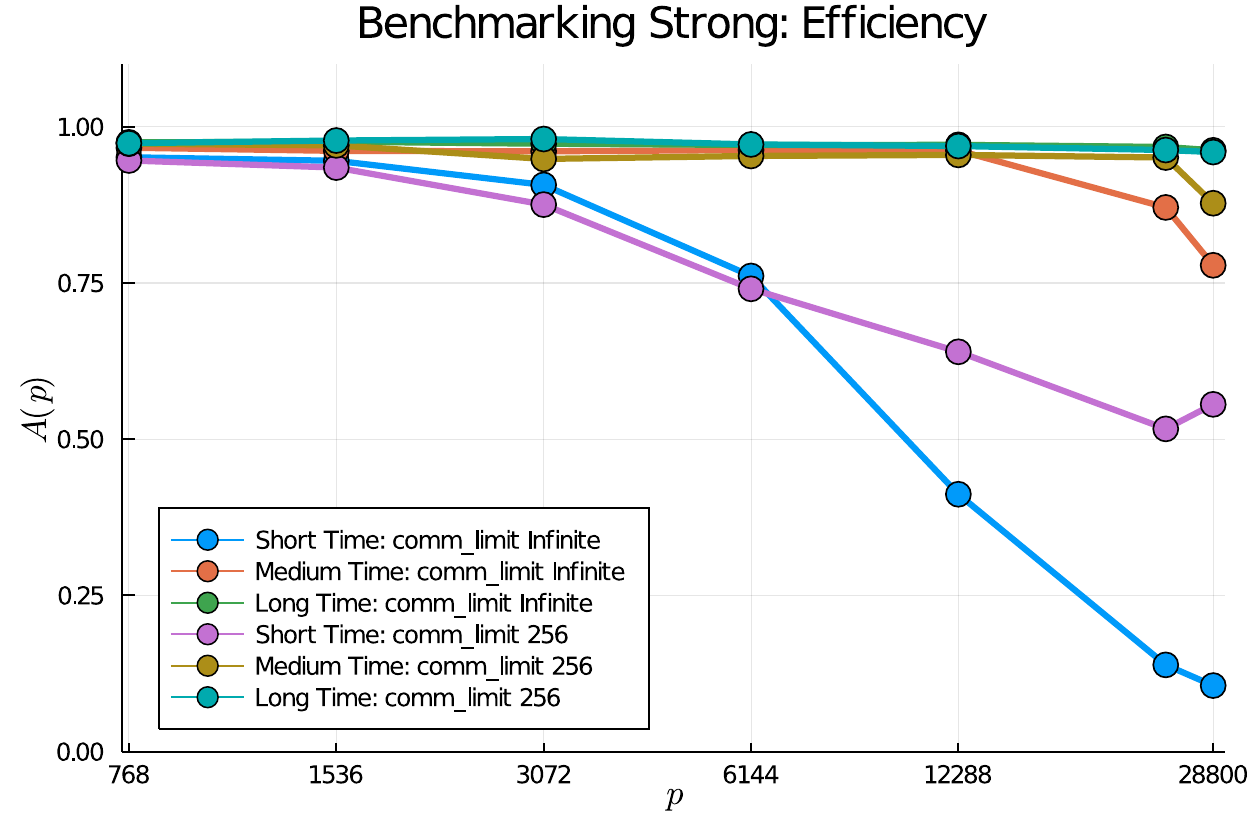}  
  \end{subfigure}
  \caption{Benchmarking efficiency under weak scaling (left) and strong scaling (right) \rev{for different sampling times (short, medium and large), with and without parallel coordinator of the scheduling}.}
  \label{fig:ParPauseNCR}
\end{figure}

\begin{figure}
  \centering
  \begin{subfigure}[b]{0.45\textwidth}
    \centering
    \includegraphics[scale=0.4]{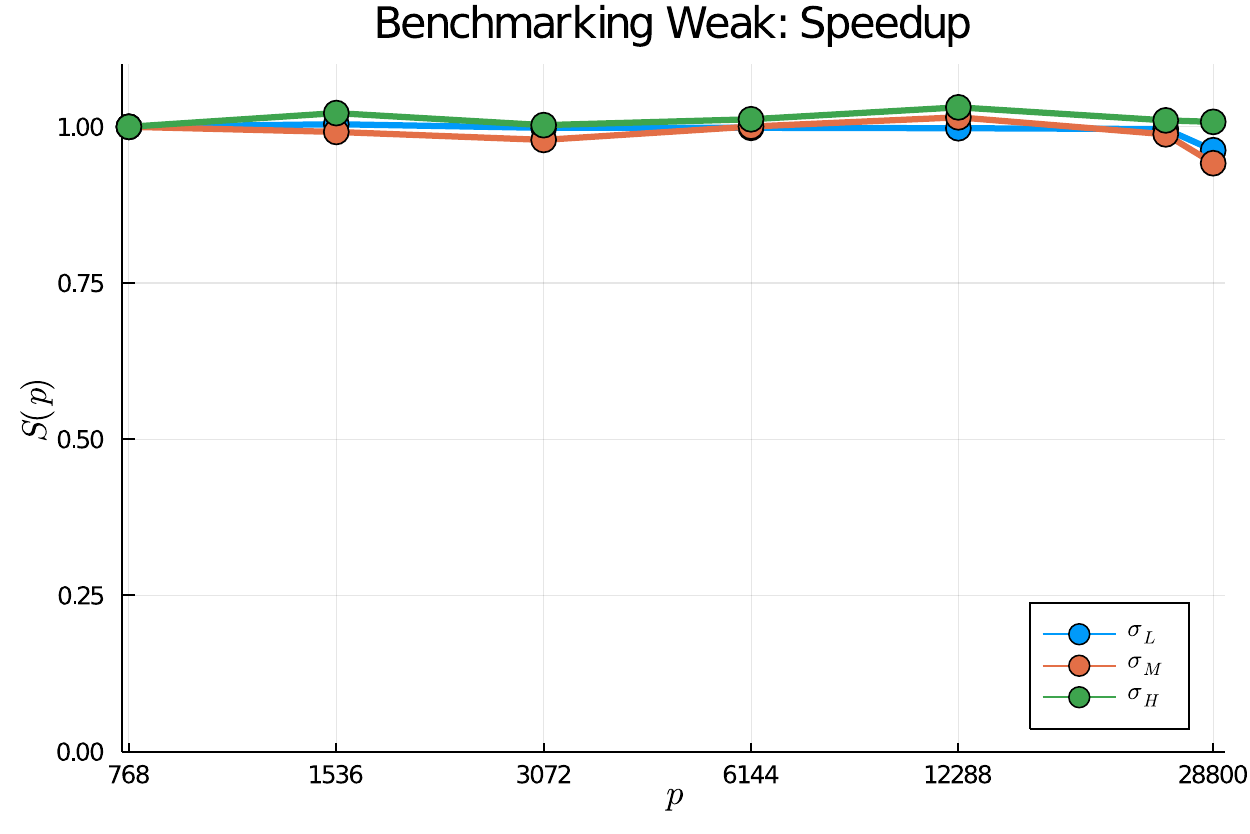}
  \end{subfigure}
  \hfill
  \begin{subfigure}[b]{0.45\textwidth}
    \centering
    \includegraphics[scale=0.4]{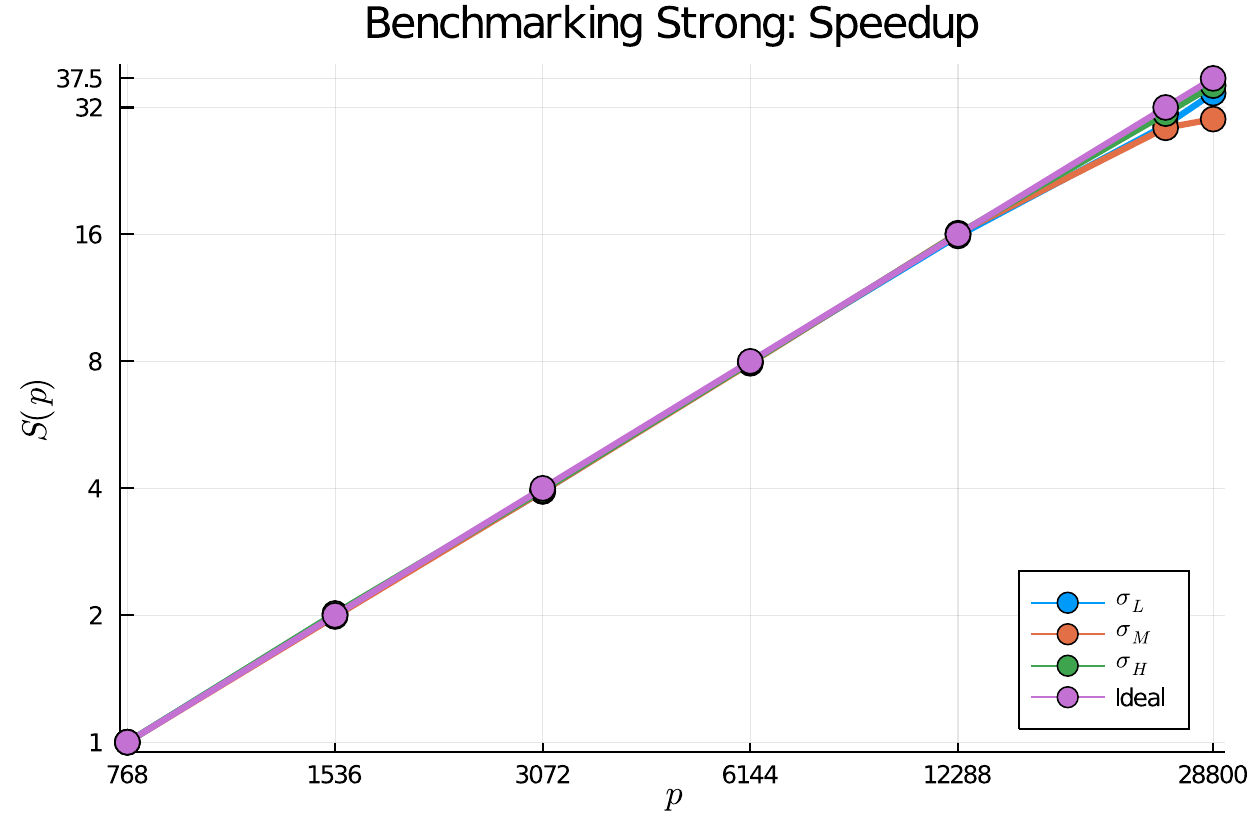}
  \end{subfigure}
  \caption{Benchmarking speedup under weak scaling (left) and strong scaling (right) \rev{for different variability in the sampling times}.}
  \label{fig:MultiSigmaScaling}
\end{figure}


\begin{figure}
  \centering
  \begin{subfigure}[b]{0.45\textwidth}
    \centering
    \includegraphics[scale=0.4]{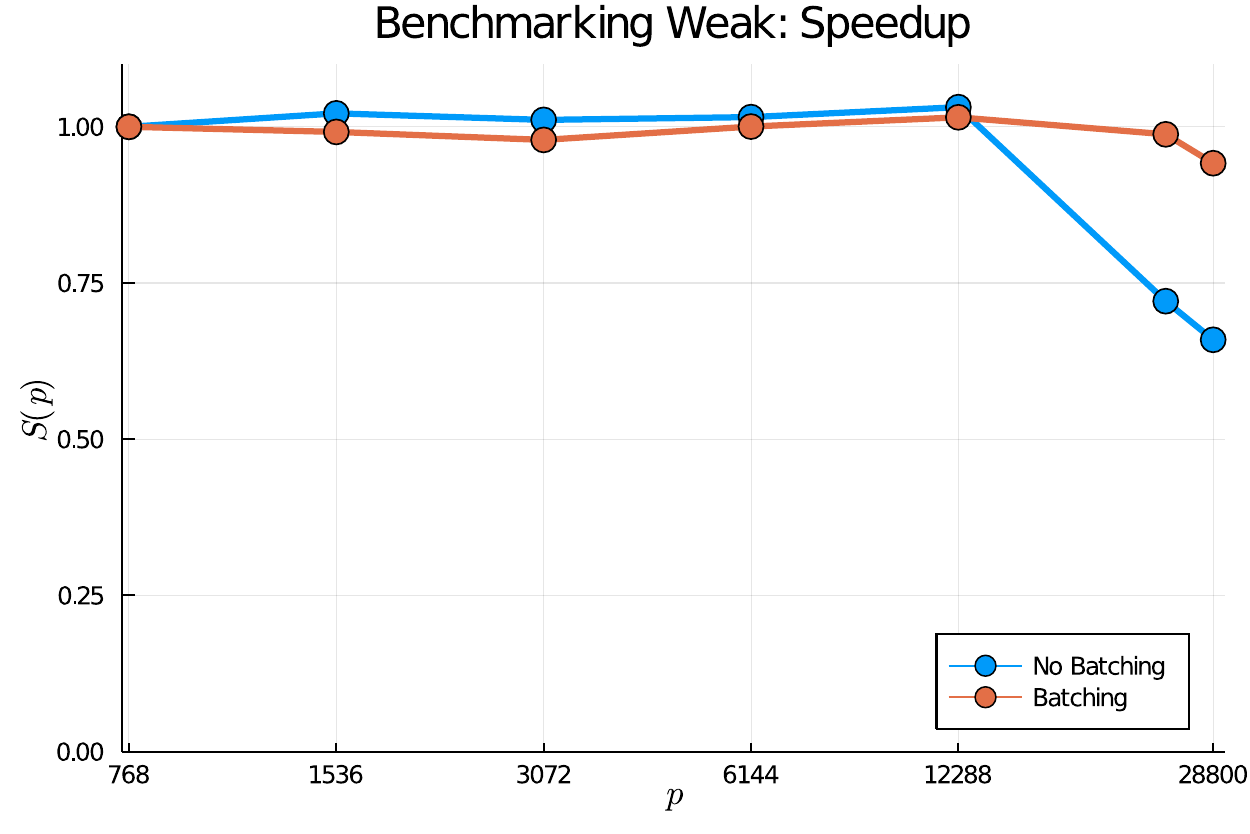}
  \end{subfigure}
  \hfill
  \begin{subfigure}[b]{0.45\textwidth}
    \centering
    \includegraphics[scale=0.4]{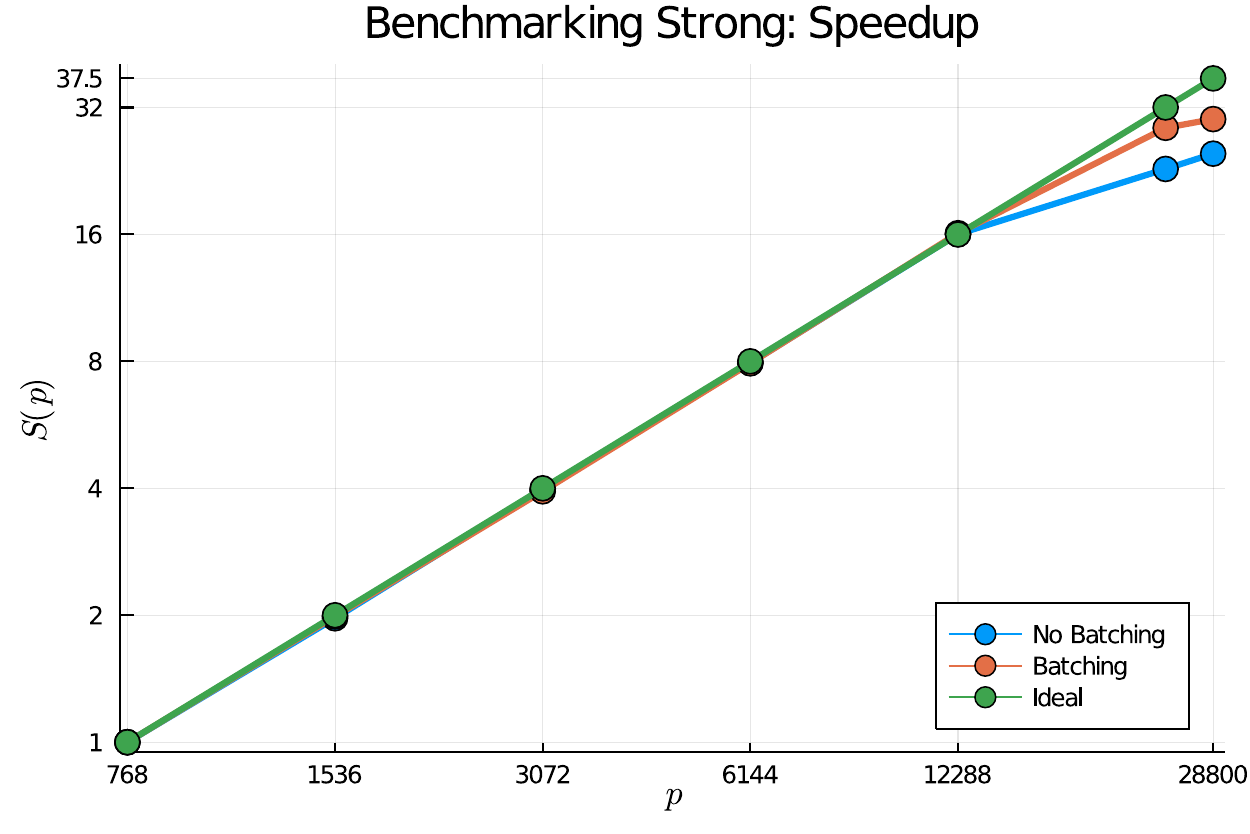}
  \end{subfigure}
  \caption{Benchmarking speedup under weak scaling (left) and strong scaling (right) \rev{with and without batch scheduling}.}
  \label{fig:NoBatchScaling}
\end{figure}


\begin{figure}
  \centering
  \begin{subfigure}[b]{0.45\textwidth}
    \centering
    \includegraphics[scale=0.4]{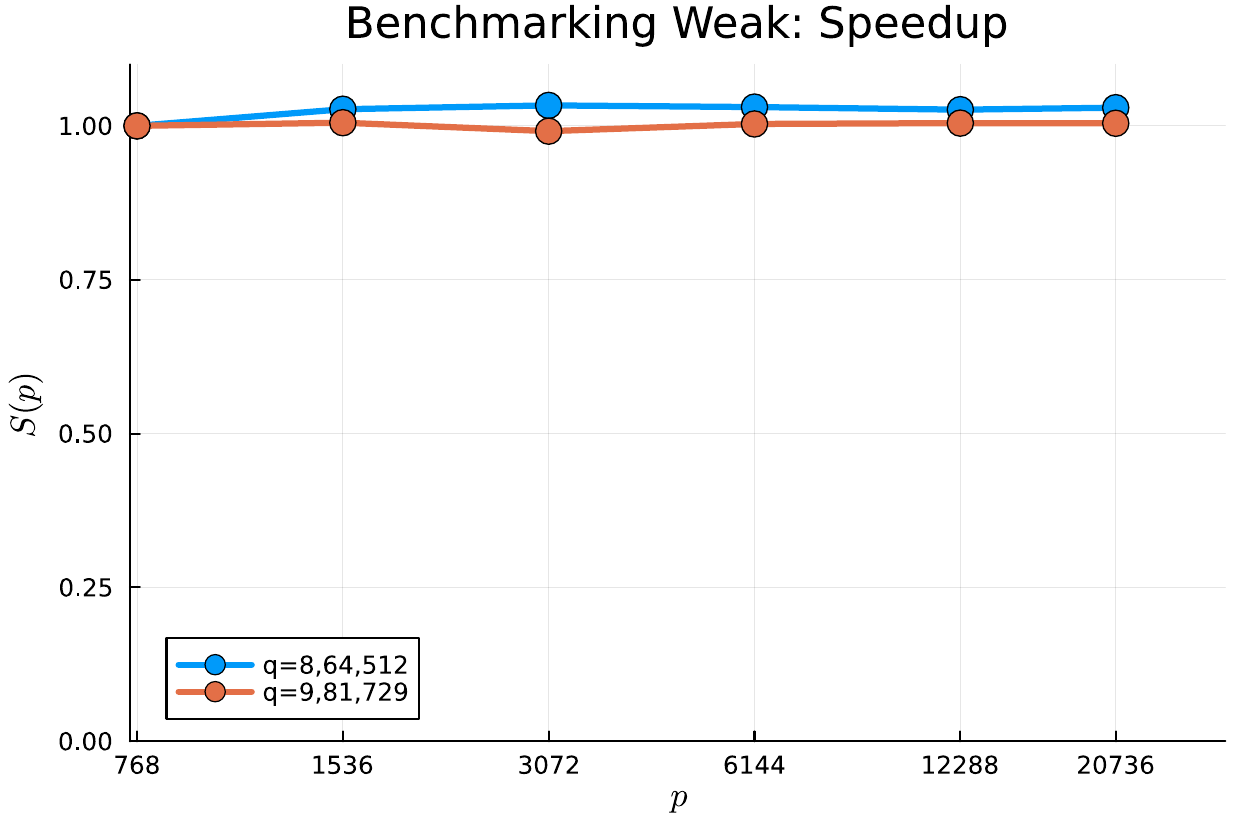}
  \end{subfigure}
  \hfill
  \begin{subfigure}[b]{0.45\textwidth}
    \centering
    \includegraphics[scale=0.4]{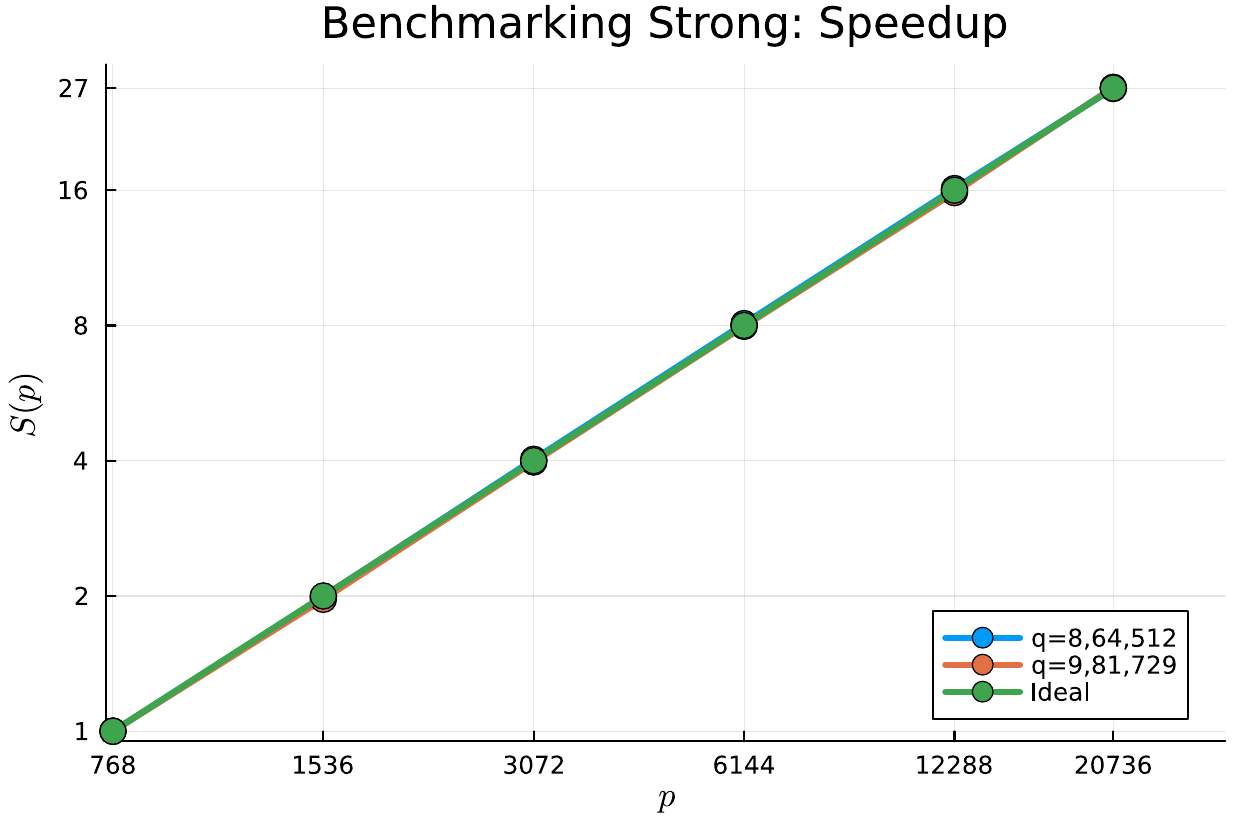}
  \end{subfigure}
  \caption{Benchmarking speedup under weak scaling (left) and strong scaling (right) for different sets of processors.}
  \label{fig:MultiProcs}
\end{figure}

\subsection{3D Poisson in a random domain}
\label{sec:Popcorn}
Here we consider a discretized finite element model, where the number of tasks per model evaluation depends on the level. To this end, we describe the details of a 3D Poisson problem as in \Eq{model}.

Here we assume the domain to be a stochastic "popcorn" shape as defined in \cite{Badia2021}, which provides a complex geometry that changes for each realized sample. 
\rev{This geometry is described by a relatively large and variable number of random variables, i.e. stochastic dimension is variable. It is composed of $n$ small ellipses 
added on top of a large one, with $n \in \mathcal{P}(11)$, i.e. a Poisson distributed random variable with mean value 11, requiring a total of $5n+7$ random variables to 
define the final shape. This gives a mean of $62$ random variables and the extra complexity of changing between samples.}

Moreover, the diffusion term $\kappa(\mathbf{x},\omega)$, is a \ac{KLE} where every point is log-normally distributed, and the associated normal distributed random variables are correlated by an exponential covariance function with correlation length $0.25$. The normal random variables have $0$ mean, and a standard deviation of $0.25$, with the first $16$ eigenfunctions in each dimension, used to construct the approximation with a total-order basis of order $16$. Here, total-order $16$ refers to a tensor product of polynomials in each dimension such that the sum of the orders of those polynomials is at most $16$. \rev{This gives a total of total of $680$ functions,
which require an equal number of random variable coefficients.}

The boundary conditions are given by a $\sin|\mathbf{x}|$ \rev{and the quantity of interest is the mean value of the solution in a square region inside the domain, see \cite{Badia2021}
for details.} \rerev{The linear system arising from the discretization of the \ac{pde} is solved using a
parallel algebraic multigrid method \cite{Henson2002} provided by HYPRE library \cite{hypre2006} through the PETSC interface.}

We consider $L=2$ \rev{(or $L_0=1$ in the case of \ac{amlmc})} levels of computation with $q_0=8, q_1=64$, and $q_2=512$. Here we consider a uniform coarsest grid of $16\times 16\times16$ elements, and each subsequent grid has twice the number of elements in each dimension, up to $64\times 64\times 64$ elements for $\lev=2$. The \ac{amlmc} here has initial samples on all three levels as given by \Eq{sample_scaling_constant} and \Tab{example_samples} for the appropriate $C$ for $\{16,32,64,128,256,512,600\}$ nodes. For strong scaling, we consider the same model parameters under strong scaling, with the sample multiplier $C$ fixed to $128$, while for weak scaling $C$ is the number of computational nodes.  

\Fig{PopcornScaling} shows the speedup for the two scaling regimes given by \Eq{speed_strong}.
We note that there are no noticeable issues scaling up to the $600$ nodes tested for this \ac{mlmc} example.

Finally, we utilize the \ac{amlmc} algorithm with this model, to test the effect of the synchronization points which that algorithm requires. For weak scaling we consider a set of desired tolerances for desired error, in this case of the volume of the associated geometric shape. Specifically, the error tolerances are chosen to be between $10^{-4}$ and $1.617\cdot10^{-5}$. These tolerances are chosen to allow the same initial sampling of the other 3D Poisson example, but such that those samples are insufficient to converge in the initial iteration. The strong scaling utilizes $C=128$, and an error tolerance corresponding of $3.535\cdot10^{-5}$.

We note that as the ultimate sample sizes for this experiment are variable, the performance is more variable. Of note, each execution leads to a different final achieved error. We consider a raw scaling, and also consider a scaling that mitigates this discrepancy, by scaling the results inversely proportional to the achieved error. That is, lower achieved errors will have used more samples, and thus more computational resources. We adjust the wall times to these values, so that the results are neutral with respect to the achieved error of the final \ac{amlmc} estimation.  The results are presented in \Fig{AMLMCScaling}, and \Fig{AMLMCNCR}. We note that both the strong and weak scaling are significantly more efficient for smaller $p$ for this problem, but that this is not a result of significant idling, so that the scalability is reasonable as $p$ increases.

\begin{figure}
  \centering
  \begin{subfigure}[b]{0.45\textwidth}
    \centering
    \includegraphics[scale=0.4]{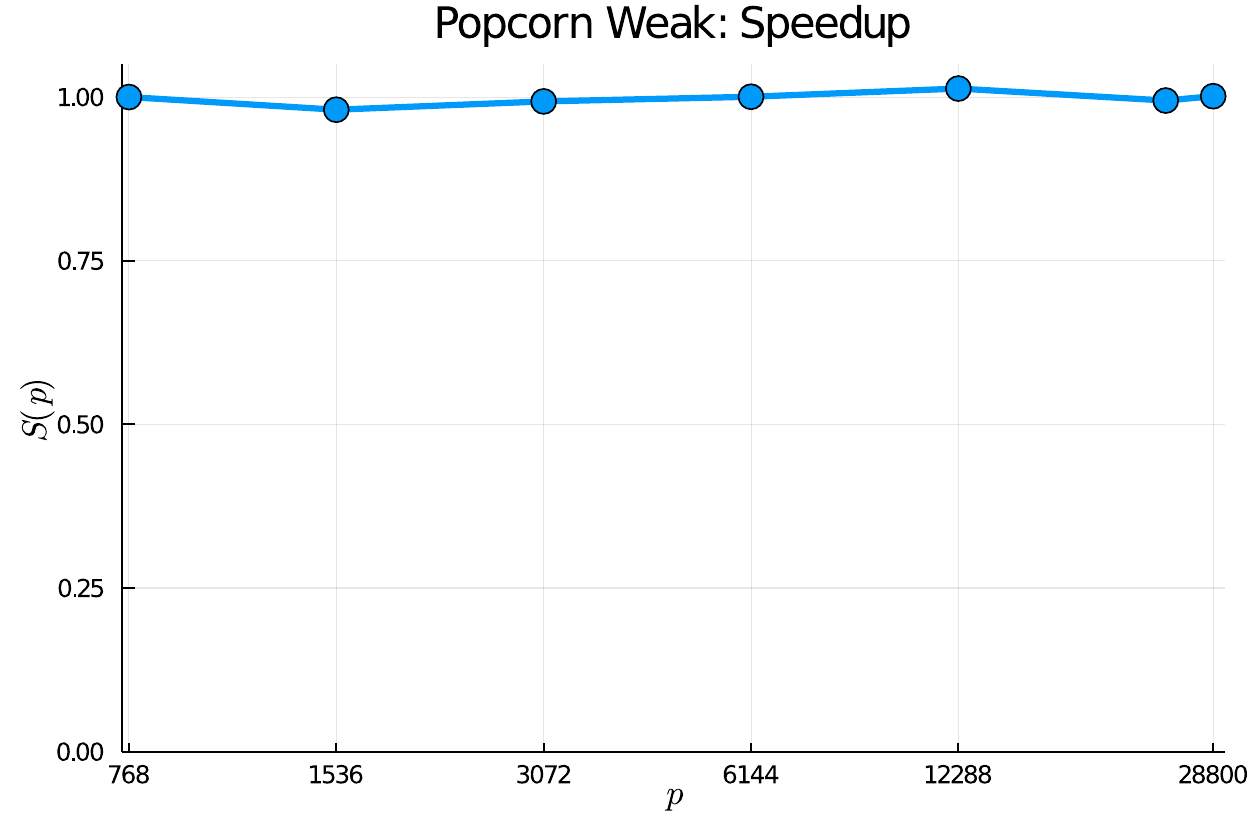}
  \end{subfigure}
  \hfill
  \begin{subfigure}[b]{0.45\textwidth}
    \centering
    \includegraphics[scale=0.4]{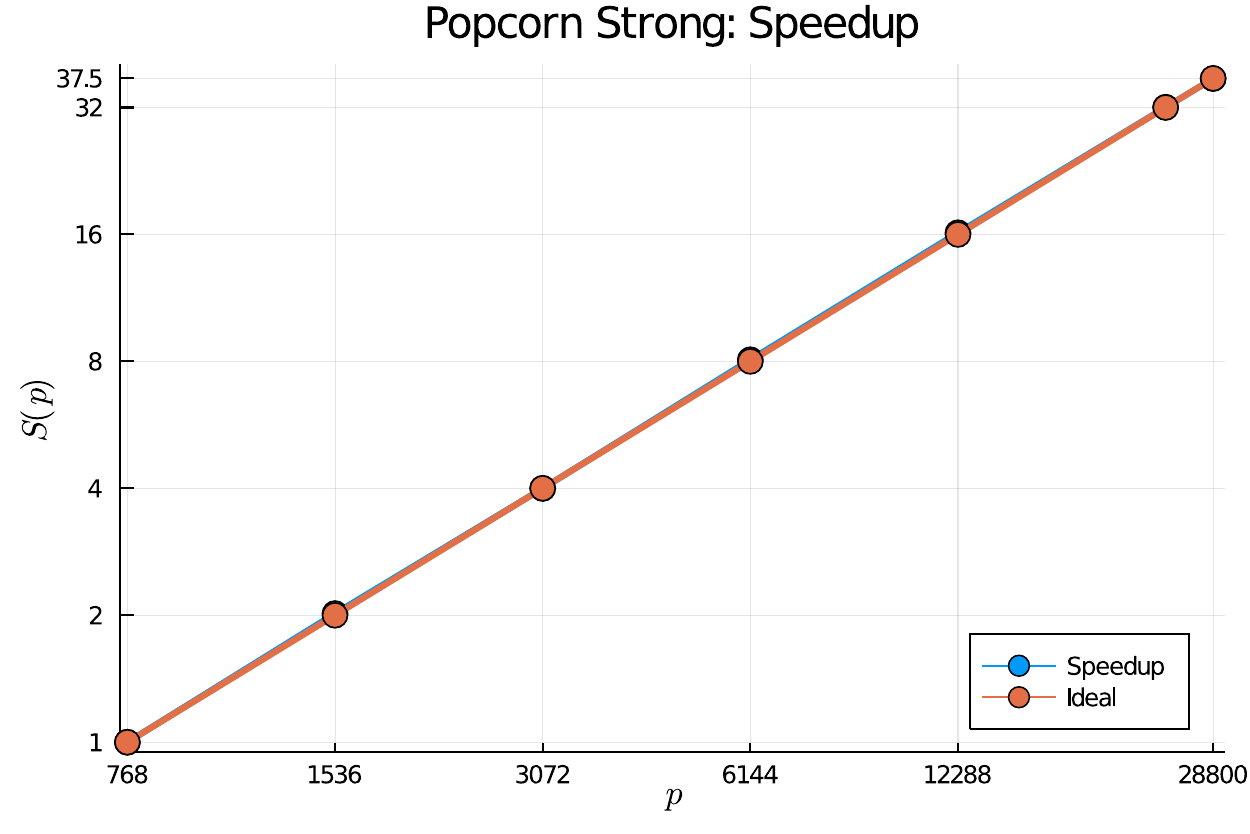}
  \end{subfigure}
  \caption{3D Poisson Popcorn speedup under weak scaling (left) and strong scaling (right).}
  \label{fig:PopcornScaling}
\end{figure}


\begin{figure}
  \centering
  \begin{subfigure}[b]{0.45\textwidth}
    \centering
    \includegraphics[scale=0.4]{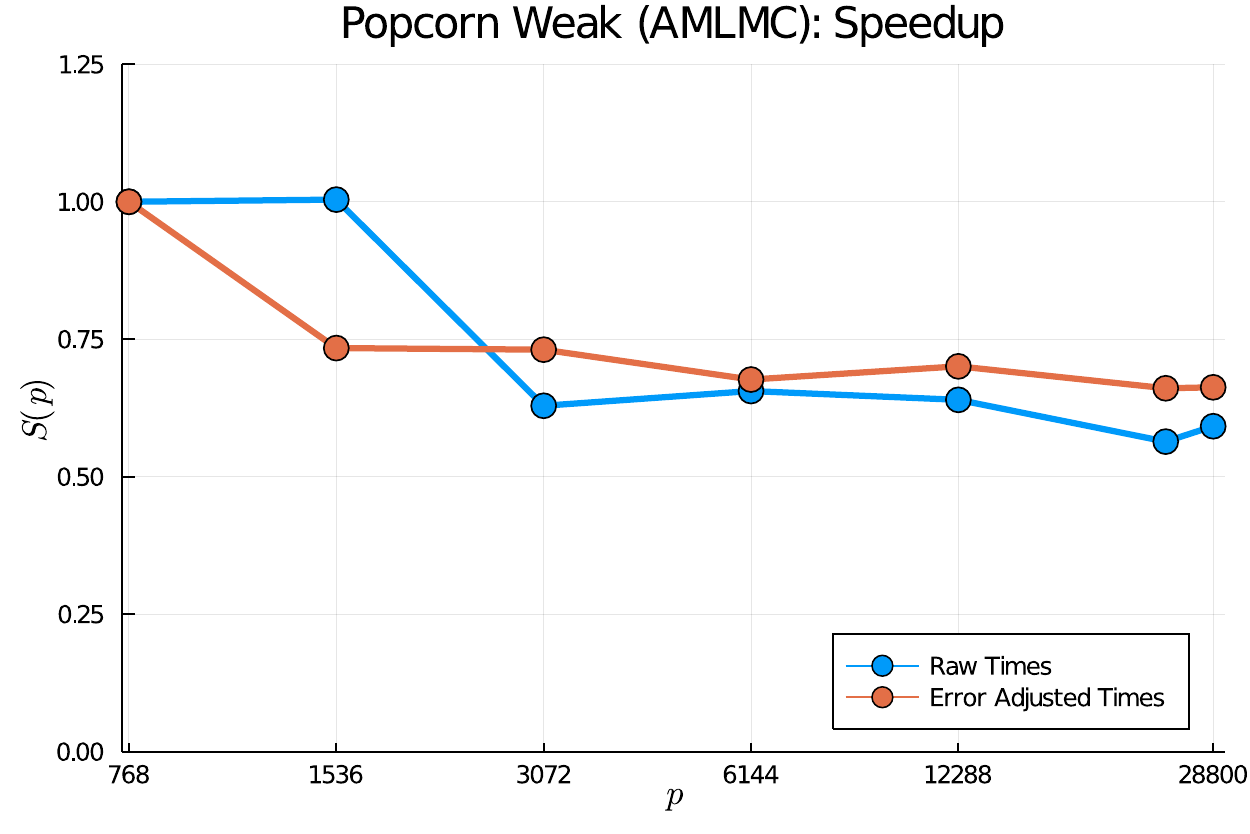}
  \end{subfigure}
  \hfill
  \begin{subfigure}[b]{0.45\textwidth}
    \centering
    \includegraphics[scale=0.4]{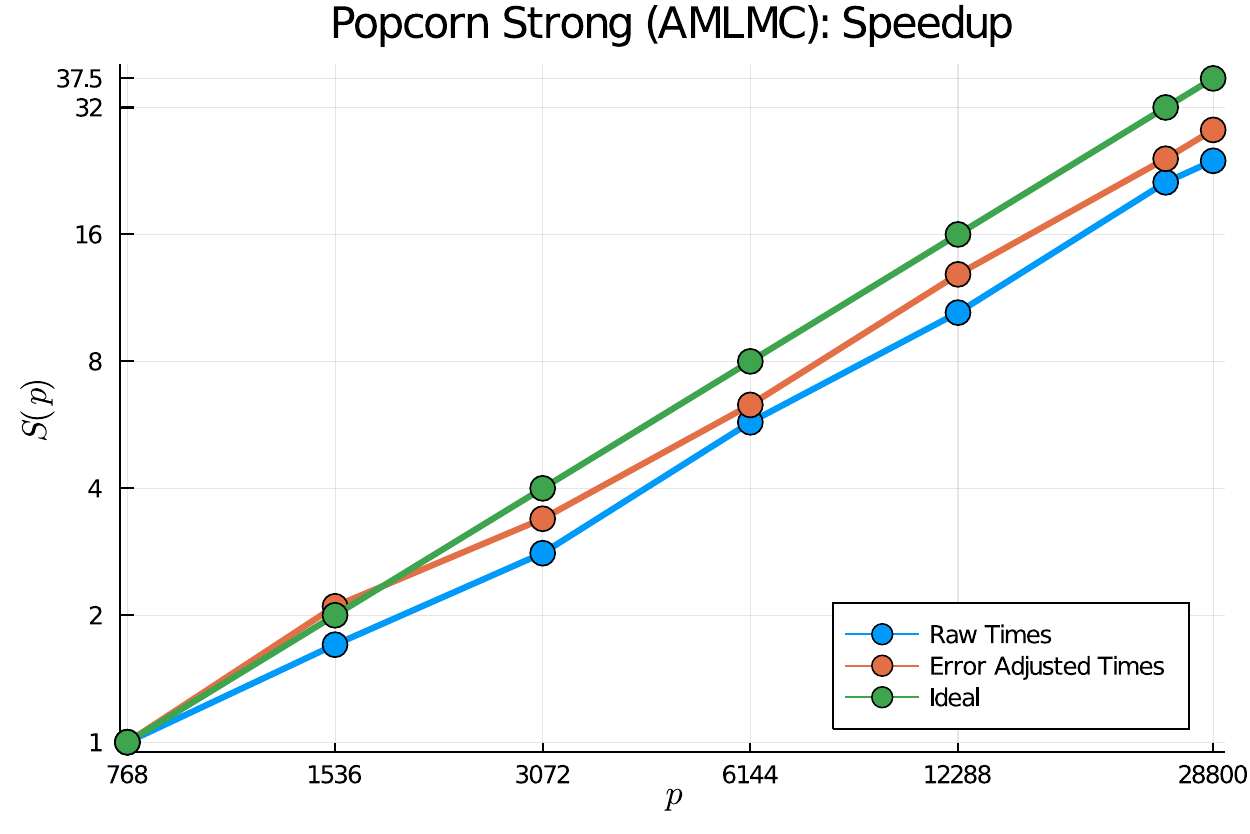}
  \end{subfigure}
  \caption{3D Poisson Popcorn \ac{amlmc} speedup under weak scaling (left) and strong scaling (right).}
  \label{fig:AMLMCScaling}
\end{figure}

\begin{figure}
  \centering
  \begin{subfigure}[b]{0.45\textwidth}
    \centering
    \includegraphics[scale=0.4]{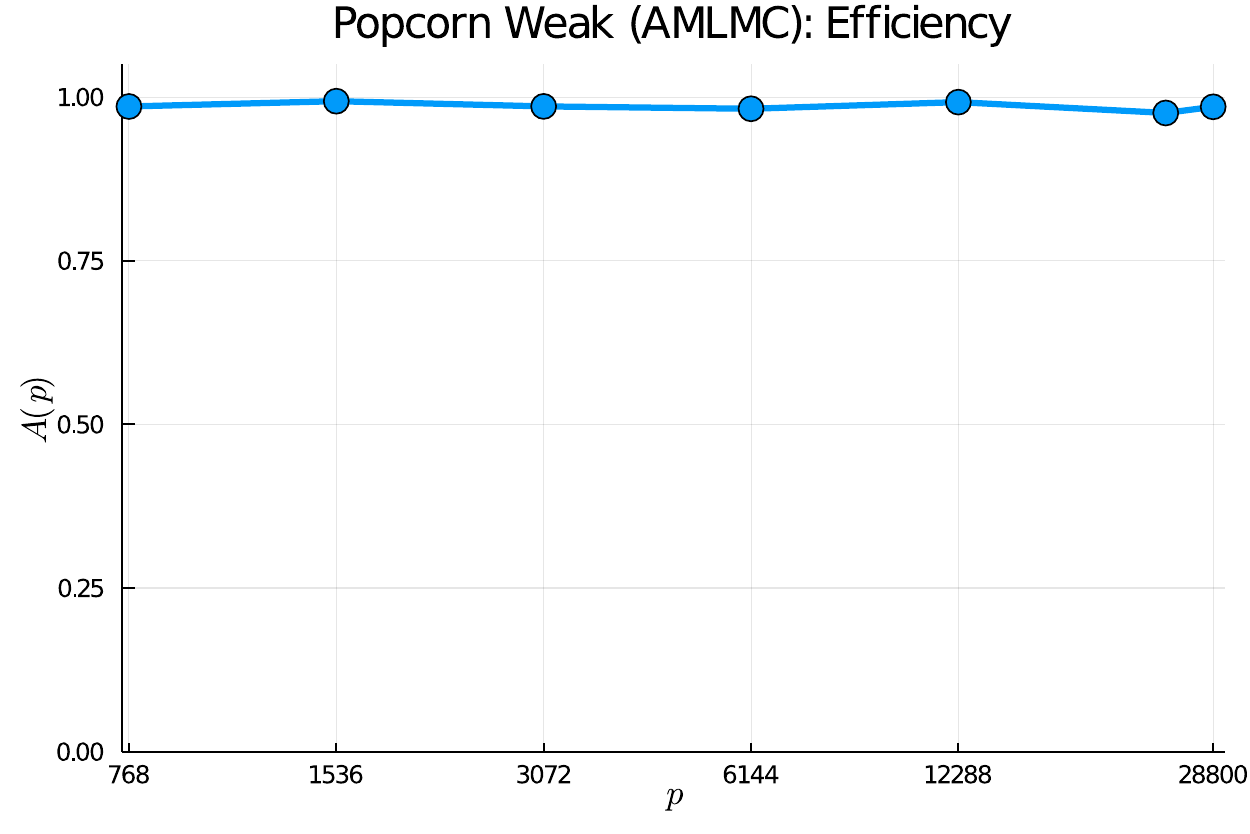}  
  \end{subfigure}
  \hfill
  \begin{subfigure}[b]{0.45\textwidth}
    \centering
    \includegraphics[scale=0.4]{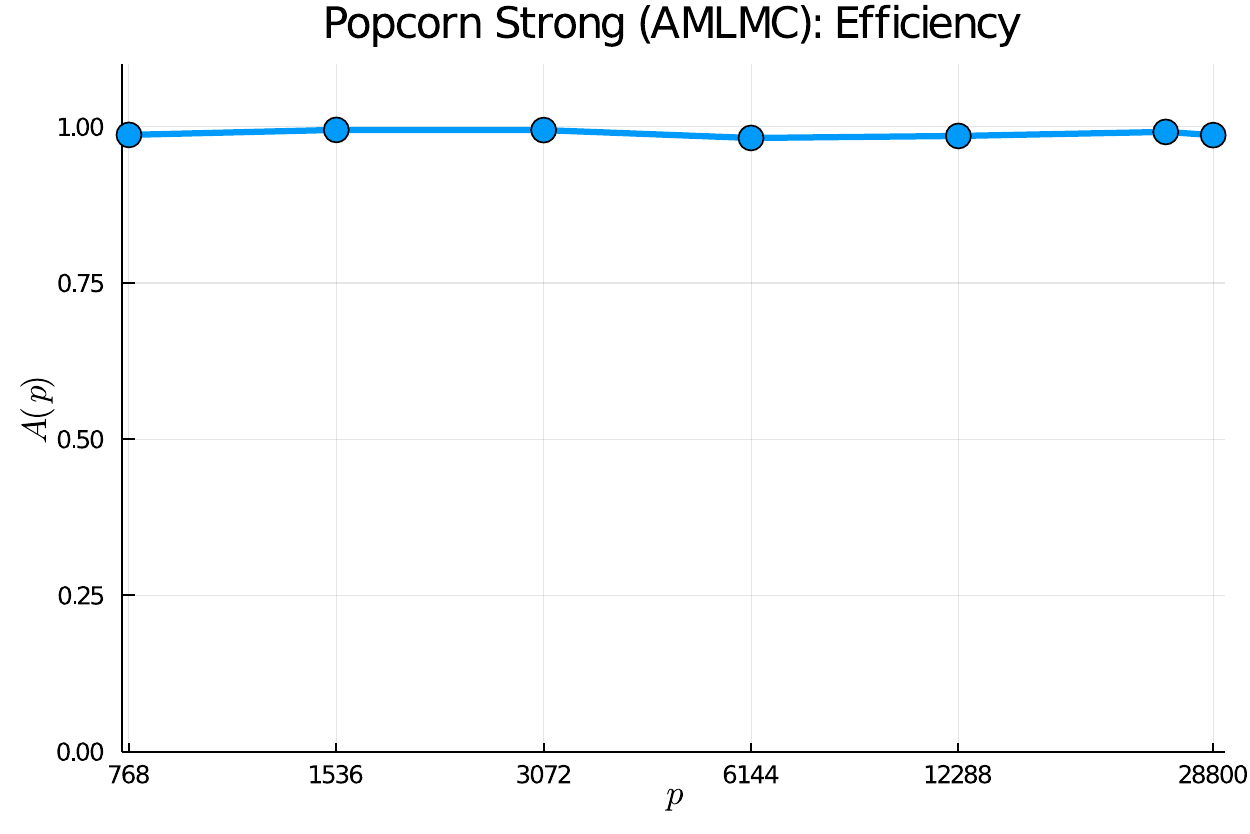}  
  \end{subfigure}
  \caption{3D Poisson Popcorn \ac{amlmc} efficiency under weak scaling (left) and strong scaling (right).}
  \label{fig:AMLMCNCR}
\end{figure}

\section{Conclusions}\label{sec:conclusions}

In this work we present a new dynamic scheduling algorithm to compute samples in the \ac{mlmc} method that fully exploits three levels of parallelism, within each sample, across samples, and across levels. The implementation of this algorithm based on the \ac{mpi} standard described here is based on general primitives available in any version of the standard and therefore in any \ac{mpi} distribution, i.e. it does not require any advanced feature like one-sided communication.

We also demonstrate the performance of method when applied to \ac{mlmc}, on a stress test of the scheduling algorithm on a benchmarking problem, as well as a $3$D Poisson model solved on a complex, sample dependent domain via \ac{amg}. The implementation is a powerful, robust, and scalable means to perform \ac{uq} via \ac{mlmc} or other similar statistical computations. The numerical examples presented herein show excelent scalability on the examples tested, when certain good scaling features are utilized. The algorithm implementation also performs well when utilized for \ac{amlmc}. Overall, our results show performance that competes with that demonstrated in the similar experiments in \cite{Gmeiner2016} and \cite{Tosi2021} (which only show results exploiting two layers of parallelism), showing high efficiency, and notable performance on the difficult case of rapidly executing models.

The results in this work could also be extended by exploiting the techniques in \cite{Tejedor2017} to account for the dependency between tasks in the context of parallel sampling. This is particularly relevant for the \ac{amlmc} method where some tasks have a dependency on other tasks that require more processors. Although this can be handled just by an iterative loop as descibed above, a dynamic scheduling with task dependency could potentially be more efficient. This development require to define appropriate priorities taking into account the amount of parallelism at sample level and is left for a future work.


%% file: badia_hampton_principe_parallel_mlmc_2021_last.bbl
\begin{thebibliography}{10}

\bibitem{Adcock2017}
Ben Adcock, Anders~C. Hansen, Clarice Poon, and Bogdan Roman.
\newblock {Breaking the coherence barrier: A new theory for compressed
  sensing}.
\newblock {\em Forum of Mathematics, Sigma}, 5, 2017.

\bibitem{Babuska2010}
Ivo Babu{\v{s}}ka, Fabio Nobile, and Ra{\'{u}}l Tempone.
\newblock {A stochastic collocation method for elliptic partial differential
  equations with random input data}.
\newblock {\em SIAM Review}, 52(2):317--355, 2010.

\bibitem{Badia2021}
Santiago Badia, Jerrad Hampton, and Javier Principe.
\newblock {Embedded multilevel Monte Carlo for uncertainty quantification in
  random domains}.
\newblock {\em International Journal for Uncertainty Quantification},
  11(1):119--142, 2021.

\bibitem{Badia2020}
Santiago Badia and Alberto~F. Mart{\'{i}}n.
\newblock {A tutorial-driven introduction to the parallel finite element
  library FEMPAR v1.0.0}.
\newblock {\em Computer Physics Communications}, 248:107059, mar 2020.

\bibitem{Badia.etal.ACME.2013}
Santiago Badia, Alberto~F. Mart{\'{i}}n, and Javier Principe.
\newblock {Implementation and Scalability Analysis of Balancing Domain
  Decomposition Methods}.
\newblock {\em Archives of Computational Methods in Engineering},
  20(3):239--262, 2013.

\bibitem{badia_multilevel_2016}
Santiago Badia, Alberto~F. Mart{\'{i}}n, and Javier Principe.
\newblock {Multilevel balancing domain decomposition at extreme scales}.
\newblock {\em SIAM Journal on Scientific Computing}, 38(1):C22--C52, 2016.

\bibitem{badia_fempar:_2017}
Santiago Badia, Alberto~F. Mart{\'{i}}n, and Javier Principe.
\newblock {FEMPAR: An Object-Oriented Parallel Finite Element Framework}.
\newblock {\em Archives of Computational Methods in Engineering},
  25(2):195--271, oct 2018.

\bibitem{Badia2017a}
Santiago Badia, Francesc Verdugo, and Alberto~F. Mart{\'{i}}n.
\newblock {The aggregated unfitted finite element method for elliptic
  problems}.
\newblock {\em Computer Methods in Applied Mechanics and Engineering},
  336:533--553, jul 2018.

\bibitem{petsc-web-page}
Satish Balay, Shrirang Abhyankar, Mark~F. Adams, Steven Benson, Jed Brown,
  Peter Brune, Kris Buschelman, Emil~M. Constantinescu, Lisandro Dalcin, Alp
  Dener, Victor Eijkhout, Jacob Faibussowitsch, William~D. Gropp, V\'{a}clav
  Hapla, Tobin Isaac, Pierre Jolivet, Dmitry Karpeev, Dinesh Kaushik,
  Matthew~G. Knepley, Fande Kong, Scott Kruger, Dave~A. May, Lois~Curfman
  McInnes, Richard~Tran Mills, Lawrence Mitchell, Todd Munson, Jose~E. Roman,
  Karl Rupp, Patrick Sanan, Jason Sarich, Barry~F. Smith, Stefano Zampini, Hong
  Zhang, Hong Zhang, and Junchao Zhang.
\newblock {PETS}c {W}eb page, 2023.

\bibitem{Barth2011}
Andrea Barth, Christoph Schwab, and Nathaniel Zollinger.
\newblock {Multi-level Monte Carlo Finite Element method for elliptic PDEs with
  stochastic coefficients}.
\newblock {\em Numerische Mathematik}, 119(1):123--161, sep 2011.

\bibitem{Baumgarten2021}
Niklas Baumgarten and Christian Wieners.
\newblock {The parallel finite element system M++ with integrated multilevel
  preconditioning and multilevel Monte Carlo methods}.
\newblock {\em Computers and Mathematics with Applications}, 81:391--406, 2021.

\bibitem{Blazewicz2019}
Jacek Blazewicz, Klaus~H. Ecker, Erwin Pesch, G{\"{u}}nter Schmidt, Malgorzata
  Sterna, and Jan Weglarz.
\newblock {\em {Handbook on Scheduling}}.
\newblock Springer International Publishing, Cham, 2019.

\bibitem{Chaudhry20181127}
Jehanzeb~H. Chaudhry, Nathanial Burch, and Donald Estep.
\newblock {Efficient distribution estimation and uncertainty quantification for
  elliptic problems on domains with stochastic boundaries}.
\newblock {\em SIAM-ASA Journal on Uncertainty Quantification},
  6(3):1127--1150, jan 2018.

\bibitem{Chen2017}
Peng Chen.
\newblock {Sparse quadrature for high-dimensional integration with Gaussian
  measure}.
\newblock {\em ESAIM: Mathematical Modelling and Numerical Analysis},
  52(2):631--657, mar 2018.

\bibitem{Cliffe2011}
K.~A. Cliffe, M.~B. Giles, R.~Scheichl, and A.~L. Teckentrup.
\newblock {Multilevel Monte Carlo methods and applications to elliptic PDEs
  with random coefficients}.
\newblock {\em Computing and Visualization in Science}, 14(1):3--15, 2011.

\bibitem{Collier2015}
Nathan Collier, Abdul~Lateef Haji-Ali, Fabio Nobile, Erik von Schwerin, and
  Ra{\'{u}}l Tempone.
\newblock {A continuation multilevel Monte Carlo algorithm}.
\newblock {\em BIT Numerical Mathematics}, 55(2):399--432, 2015.

\bibitem{Dambrine2016921}
M.~Dambrine, I.~Greff, H.~Harbrecht, and B.~Puig.
\newblock {Numerical solution of the Poisson equation on domains with a thin
  layer of random thickness}.
\newblock {\em SIAM Journal on Numerical Analysis}, 54(2):921--941, jan 2016.

\bibitem{Dambrine2017943}
M.~Dambrine, I.~Greff, H.~Harbrecht, and B.~Puig.
\newblock {Numerical solution of the homogeneous Neumann boundary value problem
  on domains with a thin layer of random thickness}.
\newblock {\em Journal of Computational Physics}, 330:943--959, feb 2017.

\bibitem{de2020touu}
Subhayan De, Jerrad Hampton, Kurt Maute, and Alireza Doostan.
\newblock {Topology optimization under uncertainty using a stochastic
  gradient-based approach}.
\newblock {\em Structural and Multidisciplinary Optimization},
  62(5):2255--2278, 2020.

\bibitem{Diaz2018b}
Paul Diaz, Alireza Doostan, and Jerrad Hampton.
\newblock {Sparse polynomial chaos expansions via compressed sensing and
  D-optimal design}.
\newblock {\em Computer Methods in Applied Mechanics and Engineering},
  336:640--666, 2018.

\bibitem{Gmeiner2016}
D.~Drzisga, B.~Gmeiner, U.~R{\"{u}}de, R.~Scheichl, and B.~Wohlmuth.
\newblock {Scheduling Massively Parallel Multigrid for Multilevel Monte Carlo
  Methods}.
\newblock {\em SIAM Journal on Scientific Computing}, 39(5):S873--S897, jan
  2017.

\bibitem{Eldred1998}
M.~Eldred and W.~Hart.
\newblock {Design and implementation of multilevel parallel optimization on the
  Intel teraflops}.
\newblock In {\em 7th AIAA/USAF/NASA/ISSMO Symposium on Multidisciplinary
  Analysis and Optimization}, pages 44--54, St. Louis, MO, sep 1998. American
  Institute of Aeronautics and Astronautics.

\bibitem{Eldred2000}
M.~Eldred, W.~Hart, B.~Schimel, and B.~Waanders.
\newblock {Multilevel parallelism for optimization on MP computers - Theory and
  experiment}.
\newblock In {\em 8th Symposium on Multidisciplinary Analysis and
  Optimization}, Long Beach, CA, sep 2000. American Institute of Aeronautics
  and Astronautics.

\bibitem{Elfverson}
Daniel Elfverson, Fredrik Hellman, and Axel Malqvist.
\newblock {A multilevel Monte Carlo method for computing failure
  probabilities}.
\newblock {\em SIAM-ASA Journal on Uncertainty Quantification}, 4(1):312--330,
  2016.

\bibitem{hypre2006}
Robert~D. Falgout, Jim~E. Jones, and Ulrike~Meier Yang.
\newblock The design and implementation of hypre, a library of parallel high
  performance preconditioners.
\newblock In Are~Magnus Bruaset and Aslak Tveito, editors, {\em Numerical
  Solution of Partial Differential Equations on Parallel Computers}, pages
  267--294, Berlin, Heidelberg, 2006. Springer Berlin Heidelberg.

\bibitem{Gantner2016}
Robert~N. Gantner.
\newblock {A generic C++ library for multilevel quasi-Monte Carlo}.
\newblock In {\em PASC 2016 - Proceedings of the Platform for Advanced
  Scientific Computing Conference}, pages 1--12, 2016.

\bibitem{ghanem1991stochastic}
Roger~G. Ghanem and Pol~D. Spanos.
\newblock {\em {Stochastic Finite Elements: A Spectral Approach}}.
\newblock Springer-Verlag New York, 1991.

\bibitem{Giles2008}
Michael~B. Giles.
\newblock {Multilevel Monte Carlo Path Simulation}.
\newblock {\em Operations Research}, 56(3):607--617, jun 2008.

\bibitem{Giles2012}
Michael~B. Giles and Christoph Reisinger.
\newblock {Stochastic finite differences and multilevel Monte Carlo for a class
  of SPDEs in finance}.
\newblock {\em SIAM Journal on Financial Mathematics}, 3(1):572--592, jan 2012.

\bibitem{Giles2009}
Michael~B Giles and Ben~J Waterhouse.
\newblock {Multilevel quasi-Monte Carlo path simulation}.
\newblock In {\em Advanced Financial Modelling, Radon Series Comp. Appl. Math},
  volume~8, pages 165--181. de Gruyter, Berlin, 2009.

\bibitem{Graham1966}
R.~L. Graham.
\newblock {Bounds for Certain Multiprocessing Anomalies}.
\newblock {\em Bell System Technical Journal}, 45(9):1563--1581, 1966.

\bibitem{Graham1969}
R.~L. Graham.
\newblock {Bounds on Multiprocessing Timing Anomalies}.
\newblock {\em SIAM Journal on Applied Mathematics}, 17(2):416--429, mar 1969.

\bibitem{Hampton2015}
Jerrad Hampton and Alireza Doostan.
\newblock {Compressive sampling of polynomial chaos expansions: Convergence
  analysis and sampling strategies}.
\newblock {\em Journal of Computational Physics}, 280:363--386, 2015.

\bibitem{Harbrecht2016}
H.~Harbrecht, M.~Peters, and M.~Siebenmorgen.
\newblock {Analysis of the domain mapping method for elliptic diffusion
  problems on random domains}.
\newblock {\em Numerische Mathematik}, 134(4):823--856, 2016.

\bibitem{Harbrecht2013}
Helmut Harbrecht and Jingzhi Li.
\newblock {First order second moment analysis for stochastic interface problems
  based on low-rank approximation}.
\newblock {\em Mathematical Modelling and Numerical Analysis},
  47(5):1533--1552, sep 2013.

\bibitem{Harbrecht2008}
Helmut Harbrecht, Reinhold Schneider, and Christoph Schwab.
\newblock {Sparse second moment analysis for elliptic problems in stochastic
  domains}.
\newblock {\em Numerische Mathematik}, 109(3):385--414, 2008.

\bibitem{Henson2002}
Van~Emden Henson and Ulrike~Meier Yang.
\newblock Boomeramg: A parallel algebraic multigrid solver and preconditioner.
\newblock {\em Applied Numerical Mathematics}, 41(1):155--177, 2002.
\newblock Developments and Trends in Iterative Methods for Large Systems of
  Equations - in memorium Rudiger Weiss.

\bibitem{Kebaier2005}
Ahmed Kebaier.
\newblock {Statistical Romberg extrapolation: A new variance reduction method
  and applications to option pricing}.
\newblock {\em Annals of Applied Probability}, 15(4):2681--2705, nov 2005.

\bibitem{lemaitre2010spectral}
Olivier Le~Maitre and Omar~M. Knio.
\newblock {\em {Spectral Methods for Uncertainty Quantification: With
  Applications to Computational Fluid Dynamics}}.
\newblock Springer, 2010.

\bibitem{Luthen2021}
Nora L{\"{u}}then, Stefano Marelli, and Bruno Sudret.
\newblock {Sparse polynomial chaos expansions: Literature survey and
  benchmark}.
\newblock {\em SIAM-ASA Journal on Uncertainty Quantification}, 9(2):593--649,
  2021.

\bibitem{Mishra2012a}
S.~Mishra, Ch~Schwab, and J.~{\v{S}}ukys.
\newblock {Multi-level Monte Carlo finite volume methods for nonlinear systems
  of conservation laws in multi-dimensions}.
\newblock {\em Journal of Computational Physics}, 231(8):3365--3388, apr 2012.

\bibitem{Mishra2012}
S.~Mishra, Ch~Schwab, and J.~{\v{S}}ukys.
\newblock {Multilevel Monte Carlo finite volume methods for shallow water
  equations with uncertain topography in multi-dimensions}.
\newblock {\em SIAM Journal on Scientific Computing}, 34(6):B761--B784, jan
  2012.

\bibitem{Mohan2011874}
P.~Surya Mohan, Prasanth~B. Nair, and Andy~J. Keane.
\newblock {Stochastic projection schemes for deterministic linear elliptic
  partial differential equations on random domains}.
\newblock {\em International Journal for Numerical Methods in Engineering},
  85(7):874--895, feb 2011.

\bibitem{Peherstorfer2018}
Benjamin Peherstorfer, Karen Willcox, and Max Gunzburger.
\newblock {Survey of multifidelity methods in uncertainty propagation,
  inference, and optimization}.
\newblock {\em SIAM Review}, 60(3):550--591, 2018.

\bibitem{Pisaroni2017}
M.~Pisaroni, F.~Nobile, and P.~Leyland.
\newblock {A Continuation Multi Level Monte Carlo (C-MLMC) method for
  uncertainty quantification in compressible inviscid aerodynamics}.
\newblock {\em Computer Methods in Applied Mechanics and Engineering},
  326:20--50, 2017.

\bibitem{Rauhut2012}
Holger Rauhut and Rachel Ward.
\newblock {Sparse Legendre expansions via $\ell_1$-minimization}.
\newblock {\em Journal of Approximation Theory}, 164(5):517--533, 2012.

\bibitem{Shegunov2020}
Nikolay Shegunov and Oleg Iliev.
\newblock {On dynamic parallelization of multilevel Monte Carlo algorithm}.
\newblock {\em Cybernetics and Information Technologies}, 20(6):116--125, 2020.

\bibitem{Sukys2014}
Jonas {\v{S}}ukys.
\newblock {Adaptive load balancing for massively parallel multi-level Monte
  Carlo solvers}.
\newblock In {\em Lecture Notes in Computer Science (including subseries
  Lecture Notes in Artificial Intelligence and Lecture Notes in
  Bioinformatics)}, volume 8384 LNCS, pages 47--56. Springer, Berlin,
  Heidelberg, 2014.

\bibitem{Sukys2014a}
Jonas {\v{S}}ukys.
\newblock {\em {Robust multi-level Monte Carlo finite volume methods for
  systems of hyperbolic conservation laws with random input data}}.
\newblock PhD thesis, 2014.

\bibitem{Sukys2012}
Jonas {\v{S}}ukys, Siddhartha Mishra, and Christoph Schwab.
\newblock {Static load balancing for Multi-Level Monte Carlo finite volume
  solvers}.
\newblock In {\em Lecture Notes in Computer Science (including subseries
  Lecture Notes in Artificial Intelligence and Lecture Notes in
  Bioinformatics)}, volume 7203 LNCS, pages 245--254. Springer, Berlin,
  Heidelberg, 2012.

\bibitem{Tejedor2017}
Enric Tejedor, Yolanda Becerra, Guillem Alomar, Anna Queralt, Rosa~M. Badia,
  Jordi Torres, Toni Cortes, and Jes{\'{u}}s Labarta.
\newblock {PyCOMPSs: Parallel computational workflows in Python}.
\newblock {\em The International Journal of High Performance Computing
  Applications}, 31(1):66--82, jan 2017.

\bibitem{Tosi2021}
Riccardo Tosi, Ramon Amela, Rosa Badia, and Riccardo Rossi.
\newblock {A parallel dynamic asynchronous framework for Uncertainty
  Quantification by hierarchical Monte Carlo algorithms}.
\newblock {\em Journal of Scientific Computing}, 89(28), 2021.

\bibitem{Vazirani2003}
Vijay~V. Vazirani.
\newblock {\em {Approximation Algorithms}}.
\newblock Springer Berlin Heidelberg, Berlin, Heidelberg, 2003.

\bibitem{Verdugo2019}
Francesc Verdugo, Alberto~F. Mart{\'{i}}n, and Santiago Badia.
\newblock {Distributed-memory parallelization of the aggregated unfitted finite
  element method}.
\newblock {\em Computer Methods in Applied Mechanics and Engineering},
  357:112583, dec 2019.

\bibitem{Xiu2005}
Dongbin Xiu and Jan~S Hesthaven.
\newblock {High-order collocation methods for differential equations with
  random inputs}.
\newblock {\em SIAM Journal on Scientific Computing}, 27(3):1118--1139, 2005.

\bibitem{Xiu2003}
Dongbin Xiu and George~Em Karniadakis.
\newblock {Modeling uncertainty in flow simulations via generalized polynomial
  chaos}.
\newblock {\em Journal of Computational Physics}, 187(1):137--167, 2003.

\bibitem{Xiu2006}
Dongbin Xiu and Daniel~M. Tartakovsky.
\newblock {Numerical methods for differential equations in random domains}.
\newblock {\em SIAM Journal on Scientific Computing}, 28(3):1167--1185, jan
  2006.

\bibitem{Zakharov2020}
Petr Zakharov, Oleg Iliev, Jan Mohring, and Nikolay Shegunov.
\newblock {Parallel Multilevel Monte Carlo Algorithms for Elliptic PDEs with
  Random Coefficients}.
\newblock In {\em Lecture Notes in Computer Science (including subseries
  Lecture Notes in Artificial Intelligence and Lecture Notes in
  Bioinformatics)}, volume 11958 LNCS, pages 463--472, 2020.

\end{thebibliography}
